\documentclass{article}
\usepackage{latexsym}
\usepackage[pdftex]{graphicx}
\DeclareGraphicsExtensions{.pdf,.jpeg,.png}
\usepackage{epstopdf}
\usepackage{amssymb}
\usepackage{amsmath,amsthm}
\usepackage[numbers,sort&compress]{natbib}
\usepackage{color}
\usepackage{amsfonts,amssymb}
\usepackage{mathrsfs}
\usepackage{dsfont}
\usepackage{graphicx}
\usepackage{ulem}
\usepackage{color}
\usepackage{bm}
\usepackage{subfigure}
\usepackage{float}
\newcommand{\upcite}[1]{$^{\mbox{\scriptsize \cite{#1}}}$}

\theoremstyle{definition}%%%% to make the content of the theorem to display ordinarily.
\newtheorem{theo}{Theorem}[section]
\newtheorem{lem}{Lemma}[section]
\newtheorem{remark}{Remark}[section]
\begin{document}
	\title{An integrable semi-discretization of the modified Camassa-Holm equation with linear dispersion term}
	\author{Han-Han Sheng$^\dag$, Guo-Fu Yu$^\dag$ and Bao-Feng Feng$^\ddag$ \footnote{Corresponding author. Email address: baofeng.feng@utrgv.edu}\\
		$^\dag$ School of Mathematical Sciences, Shanghai Jiao Tong University, \\
		Shanghai 200240, P.R.\ China \\
	$^\ddag$ School of Mathematical and Statistical Sciences,\\
	The University of Texas Rio Grande Valley, Edinburg, Texas 78541, USA}
	\date{}
	\maketitle
\begin{abstract}
In the present paper, we are concerned with integrable discretization of a modified Camassa-Holm equation with linear dispersion term. The key of the construction is the semi-discrete analogue for a set of bilinear equations of the modified Camassa-Holm equation. Firstly, we show that these bilinear equations and their determinant solutions either in Gram-type or Casorati-type can be reduced from the discrete KP equation through Miwa transformation. Then, by scrutinizing the reduction process, we obtain a set of semi-discrete bilinear equations and their general soliton solution in Gram-type or Casorati-type determinant form. Finally, by defining dependent variables and discrete hodograph transformations, we are able to derive  an integrable semi-discrete analogue of the modified Camassa-Holm equation. It is also shown that the semi-discrete modified Camassa-Holm equation converges to the continuous one in the continuum limit.
\end{abstract}
\section{Introduction}
In this paper, we consider integrable semi-discretization of the modified Camassa-Holm (mCH) equation
\begin{align}
m_t+2\kappa^2u_x+[m(u^2-u_x^2)]_x=0,\quad m=u-u_{xx},\label{mch}
\end{align}
which firstly appeared in the papers of Fokas\upcite{Fokas}, Fuchssteiner\upcite{Fuch}, Olver and Rosenau\upcite{olv} and later, was rediscovered by Qiao\upcite{qiao,qiao1}. Here $u=u(x,t)$ is a real-valued function, and the subscripts $x$ and $t$ appended to $m$ and $u$ denote partial differentiation. The  parameter $\kappa$ characterizes the magnitude of the linear dispersion.
The mCH equation (\ref{mch}) has attracted considerable attention  in the past two decades due to its rich mathematical structure and solutions.  Several groups have studied the mCH equation from various aspects such as the well-posedness, regularization, Cauchy problem, Riemann-Hilbert problem, long-time asymptotics and Liouville correspondence with the mKdV equation \upcite{Qu3,Tang,regularization,Qu1,FORQ1,cauchy,RHP,longtime,OlverQu2016}. The smooth soliton solutions were constructed by different approaches such as the Hirota's bilinear method\upcite{mats,hu-soliton}, Darboux transformation/B\"acklund transformation method\upcite{DarbouxQiao, BacklundQP,FORQ2}.  Recently, Chang et al. considered the Lax integrability and the conservative peakon solutions in a series of work\upcite{Chang1,Chang2,Chang3}  while Gao et al. studied the patched peakon weak solution\upcite{Patched}.
Similar to the Camassa-Holm equation\upcite{CH}, the wave-breaking and blow-up phenomena are very intriguing problems which have been studied by several authors\upcite{wave-break,blow-up,oscillation}. The stability including the orbital stability of the mCH equation was analyzed in Refs. \upcite{stab,orbital}.

Recently, much attention has been paid to the study of discrete integrable systems that are related to many other fields such as
quantum field theory, numerical algorithms, random matrices, orthogonal and bi-orthogonal polynomials\upcite{Disbook}.  In compared with the continuous integrable systems,  the examples of discrete integrable systems, as well as the tools of analyzing them are much less. On the other hand, it is believed that  discrete integrable systems are more fundamental and universal than continuous ones.  Starting from the discrete Kadomtsev-Petviashvili (KP) equation, or the so-called Hirota--Miwa (HW) equation\upcite{Hirota-1981, Miwa-1982}, Shi {\it et al.} have derived discrete KdV equation and discrete potential mKdV equation, as well as their Lax pairs and multi-soliton solutions\upcite{Nimmo-2014,Nimmo-2017}. The authors have done a series of work in finding integrable discretizations of soliton equations such as the short pulse equation\upcite{Feng1,Feng2}, (2+1)-dimensional Zakharov equation\upcite{Yu}, the Camassa-Holm equation\upcite{Feng3,d-CH3} and the Degasperis-Proceli equaiton\upcite{dDP}.  Therefore, it is a natural but definitely not a trivial problem for us to construct an integrable discrete analogue of the mCH equation (\ref{mch}).

The remainder of the paper is organized as follows. In section \ref{sec2}, starting from the discrete KP equation and its tau function, we derive a set of bilinear equations that belong to the negative KdV hierarchy through a series of transformations including Miwa transformation. Moreover, by introducing a hodograph and dependent variable transformations, we show that the mCH equation
(\ref{mch}) is deduced from a set of bilinear equations in question. As a by-product, we provide the multi-soliton solution to the mCH equation in both Gram-type and Casorati-type determinant.  In section \ref{sec3},  by scrutinizing the process in deriving the mCH equation from the discrete KP equation, we propose semi-discrete analogues of a set of bilinear equations which derive the mCH equation. Based on these discrete bilinear equations, we construct an integrable semi-discrete mCH equation and present its $N$-soliton solutions.  Section \ref{sec4} is devoted to a brief summary and discussion. Some detailed proofs are given in Appendices.
	
\section{From the discrete KP equation to the mCH equation}\label{sec2}
In this section, we will show that the bilinear equations and the multi-soliton solution to the mCH equation (\ref{mch}) given by Matsuno\upcite{mats} can be generated from the discrete KP equation and its determinant solution through a series of transformations including the Miwa transformation, hodograph transformation and dependent variable transformation.
\subsection{A brief review for the discrete KP equation}
The discrete Kadomtsev-Petviashvili (KP) equation, or the so-called Hirota--Miwa (HW) equation\upcite{Hirota-1981, Miwa-1982}, is a three-dimensional discrete integrable system
\begin{gather}\label{H-M-1}
(a_1-a_2)\tau_{12}\tau_3+(a_2-a_3)\tau_{23}\tau_1+(a_3-a_1)\tau_{13}\tau_2=0,
\end{gather}
where lattice parameters $a_k$ are distinct constants, $k=1, 2, 3$, and for $\tau=\tau(k_1,k_2,k_3)$ each subscript $i$ denotes a forward shift in the corresponding discrete variable $k_i$. It was discovered by Hirota\upcite{Hirota-1981} as a fully discrete analogue of the two-dimensional Toda equation and later Miwa\upcite{Miwa-1982} showed that it was intimately related to the KP hierarchy.

The discrete KP hierarchy, or the Hirota-Miwa system, can be expressed by an infinite number of bilinear equations with $(k_i, k_j, k_m)$ which are taken from $(k_1, k_2, k_3, \cdots)$.
\begin{equation}
(a_{i}-a_{j}) {\tau }_{ij} \tau _{m} + (a_{j}-a_{m})  \tau_{jm} \tau _{i}+(a_{m}-a_{i}) {\tau }_{mi} \tau _{j}=0\,.
\end{equation}
The Hirota--Miwa equation (\ref{H-M-1}) arises as the compatibility condition of the linear system\upcite{Nimmo-1997}
\begin{gather}\label{H-M-LP-1}
\phi_i-\phi_j=(a_i-a_j)u^{ij}\phi, \qquad 1\leq i<j\leq3,
\end{gather}
where $\phi=\phi(k_i,k_j,k_m)$, $\tau=\tau(k_i,k_j,k_m)$, $u^{ij}=\tau_{ij}\tau/(\tau_i\tau_j)$. Each subscript $i$ denotes a forward shift in the corresponding discrete variable~$n_i$, for example, $\phi_{i} = \phi(k_i+1,k_j,k_m)$.

It is known that the discrete KP equation admits a general solution in terms of the following Gram-type  determinant\upcite{OHTI_JPSJ}
\[
\tau (k_{1},k_{2},k_{3})=\Big|m_{ij}(k_{1},k_{2},k_{3})\Big|_{1\leq i,j\leq
N}
\]
where
\[
m_{ij} (k_{1},k_{2},k_{3})=c_{ij}+\frac{1}{p_{i}+q_{j}} \prod_{l=1}^3 \phi(a_{l},k_{l})\bar{\phi}(a_{l},k_{l})
\]
with
\[
\phi(a_{l},k_{l})=\left(a_{l}-p_{i}\right) ^{-k_{l}}, \quad  \bar{\phi}(a_{l},k_{l})=\left( {a_{l}+q_{j}}\right) ^{k_{l}}
\]
Therefore, the Gram-type solution can be specifically expressed as
\begin{equation}
\tau (k_{1},k_{2},k_{3})=\Big|c_{ij}+\frac{1}{p_{i}+q_{j}} \left(- \frac{p_{i}-a_{1}}{q_{j}+a_{1}}\right) ^{-k_{1}} \left(- \frac{p_{i}-a_{2}}{q_{j}+a_{2}}\right) ^{-k_{2}} \left(-\frac{p_{i}-a_{3}}{q_{j}+a_{3}}\right) ^{-k_{3}}\Big|\,.
\label{HW-Gram}
\end{equation}
Notice that
\begin{eqnarray*}
\left( -\frac{p_{i}-a_{1}}{q_{j}+a_{1}}\right) ^{-k_{1}} &=&\exp \left(
-k_{1}\ln \left( \frac{1-a^{-1}_{1}p_{i}}{1+a^{-1}_{1}q_{j}}\right) \right) \\
&=&\exp \left( (p_{i}+q_{j})a^{-1}_{1}k_{1}+(p_{i}^{2}-q_{j}^{2})\frac{1}{2}%
a_{1}^{-2}k_{1}+(p_{i}^{3}+q_{j}^{3})\frac{1}{3}%
a_{1}^{-3}k_{1}+\cdots \right)\,.
\end{eqnarray*}
Thus by defining the so-called Miwa transformation
\begin{equation*}
x_{1} =\sum_{j=1}^3 a^{-1}_{j}k_{j}\,,\ \   x_{2} =\frac{1}{2}\sum_{j=1}^3 a_{j}^{-2}k_{j}, \ \ \cdots \ \ x_{n} =\frac{1}{n}\sum_{j=1}^3 a_{j}^{-n}k_{j},
\end{equation*}
and further by using elementary Schur polynomial
\[
\exp \left( \sum t^{n}x_{n}\right) =\sum p_{n}(\vec{x})t^{n}\,,\quad
\vec{x}=(x_{1},x_{2},\cdots ,x_{n}),
\]
one can obtain the whole KP hierarchy\upcite{OHTI_JPSJ}
\begin{eqnarray*}
&&\left((a_{2}-a_{3})\sum a_{1}^{-K}a_{2}^{-L}a_{3}^{-M}p_{K}( -\frac{1}{2}%
\widetilde{D}) p_{L}( \frac{1}{2}\widetilde{D}) p_{M}(
\frac{1}{2}\widetilde{D})\right. \\
&& +(a_{3}-a_{1})\sum a_{1}^{-K}a_{2}^{-L}a_{3}^{-M}p_{K}( \frac{1}{2}%
\widetilde{D}) p_{L}( -\frac{1}{2}\widetilde{D})
p_{M}( \frac{1}{2}\widetilde{D})  \\
&&\left.+(a_{1}-a_{2})\sum a_{1}^{-K}a_{2}^{-L}a_{3}^{-M}p_{K}( \frac{1}{2}%
\widetilde{D}) p_{L}( \frac{1}{2}\widetilde{D}) p_{M}(
-\frac{1}{2}\widetilde{D})\right) \tau \cdot \tau=0
\end{eqnarray*}
where $\widetilde{D}=\left(D_{x_{1}},\frac{1}{2}D_{x_{2}},\cdots ,\frac{1}{n}%
D_{x_{n}}\right)$.
If we take $(K,L,M)=(1,2,3)$, then we obtain the bilinear equation for the  KP equation
\[
(D_{x_{1}}^{4}-4D_{x_{1}}D_{x_{3}}+3D_{x_{2}}^{2})\tau \cdot \tau =0
\]
Next, let us show how we can generate fully discrete 2-dimensional Toda lattice (2DTL) and the whole 2DTL hierarchy. To this end, we reparametrize $p_{i}$ and $q_{j}$ by
\[
a_{1}-p_{i}=-\widetilde{p}_{i},a_{1}+q_{j}=\widetilde{q}_{j},
\]
so that
\begin{eqnarray*}
  &&c_{ij}+\frac{1}{p_{i}+q_{j}} \left(- \frac{p_{i}-a_{1}}{q_{j}+a_{1}}\right) ^{-k_{1}} \left(- \frac{p_{i}-a_{2}}{q_{j}+a_{2}}\right) ^{-k_{2}} \left(-\frac{p_{i}-a_{3}}{q_{j}+a_{3}}\right) ^{-k_{3}}
   \\
  &&  \rightarrow c_{ij}+ \left( -\frac{\widetilde{p}_{i}}{%
\widetilde{q}_{j}}\right) ^{-(k_{1}+k_3)}\left( \frac{1-a\widetilde{p}_{i}}{1+a%
\widetilde{q}_{j}}\right) ^{-k_{2}}\left( \frac{1-b \widetilde{p}^{-1}_{i}}{%
1+b \widetilde{q}^{-1}_{j}}\right) ^{-k_{3}},
\end{eqnarray*}
where $a^{-1}=a_2-a_1$, $b=a_3-a_1$.  Moreover, let $k_1+k_3+1=n$, $k_2=k$, $k_3=l$, then the HW equation  (\ref{H-M-1}) is converted into
\begin{eqnarray*}
&&(1 -ab)\tau_n (k,l)\tau_n (k+1,l+1)
+ab \tau_{n-1} (k+1,l)\tau_{n+1} (k,l+1)-  \tau_n (k, l+1)\tau_n (k+1,l)=0
\end{eqnarray*}
%\begin{eqnarray*}
%&&ab (\tau_n (k,l)\tau_n (k+1,l+1)-\tau_{n-1} (k+1,l)\tau_{n+1} (k,l+1)) = \tau_n (k,l)\tau_n (k+1,l+1)-\tau_n (k, l+1)\tau_n (k+1,l)
%\end{eqnarray*}
which is exactly the bilinear equation of the discrete two-dimensional Toda-lattice (2DTL) equation. By dropping the tilde, the discrete 2DTL equation admits the following general soliton
\begin{eqnarray}
&&\tau_n (k,l)= \left| c_{ij}+ \frac{1}{p_i+q_j} \left(-\frac{p_i}{q_j}\right)^{-n}   \left( \frac{1-a{p}_{i}}{1+a%
{q}_{j}}\right) ^{-k}\left( \frac{1-b {p}^{-1}_{i}}{%
1+b {q}^{-1}_{j}}\right) ^{-l}
 \right|\,.
 \label{d2DTL-Gram}
\end{eqnarray}
Upon applying the Miwa transformations in both positive and negative flows
\begin{align}
&x_{1} =ka\,,\ \  x_{2} =\frac{1}{2} ka^{2}, \ \ \cdots, \ \ x_{n} =\frac{1}{n} ka^{n},\\
&x_{-1} =lb\,,\ \   x_{-2} =\frac{1}{2} lb^{2}, \ \ \cdots, \ \ x_{-n} =\frac{1}{n} lb^{n},
\end{align}
one obtains the whole 2DTL hierarchy, among which
\begin{equation*}
\left(\frac{1}{2} D_{x_1} D_{x_{-1}}-1\right)\tau_n \cdot\tau_n
=-\tau_{n+1}\tau_{n-1}\,,
\end{equation*}
is the celebrating bilinear equation of the  2DTL equation that admits the Gram-type determinant solution
\begin{equation}
\tau_n=\left|c_{ij}+ \frac{1}{p_{i}+q_{j}} \left(-\frac{p_i}{q_j}\right)^{-n} e^{\xi_i+\eta_j} \right|,
\end{equation}
with
\[
\xi_i= \frac{1}{p_i} x_{-1} + p_i x_1 +  \cdots, \quad \eta_j= \frac{1}{q_j} x_{-1} + q_j x_1 +  \cdots\,.
\]
\subsection{From the discrete KP equation to the modified CH equation}
First, we give a lemma regarding bilinear equations of the mCH equation (\ref{mch}).
\begin{lem}
The following bilinear equations
\begin{align}
&(D_{x_1}+a)\tau_{n,k+1}\cdot \tau_{n+1,k}=a\tau_{n+1,k+1}\tau_{n,k},\label{BL1}\\
&(D_{x_{-1}}(D_{x_1}+a)-2)\tau_{n,k+1}\cdot\tau_{n+1,k}=(aD_{x_{-1}}-2)\tau_{n+1,k+1}\cdot \tau_{n,k},\label{BL2}
\end{align}
is generated from the discrete KP equation (\ref{H-M-1}). Moreover, above two bilinear equations admit
determinant solution of Gram-type
\begin{equation}
		\tau_{n,k}=\left| m^{n,k}_{ij}\right|=\left| c_{ij}+ \left(-\frac{{p}_{i}}{{q}_{j}}\right)^{-n}\left(-\frac{{p}_{i}+a}{{q}_{j}-a}\right)^{-k} e^{\xi_i+\eta_j} \right|,
		\label{Gram}
\end{equation}
where
\begin{align*}
	\xi_i=p_ix_1+\frac{1}{p_i+a}x_{-1}+\xi_{i0},\,\ \eta_i=q_i x_1+\frac{1}{q_i-a}x_{-1}+\eta_{i0}.
\end{align*}
Here $c_{ij}, a, p_i,q_i,\xi_{i0},\eta_{i0}$ are the arbitrary parameters.
\end{lem}
\begin{proof}
Notice that the element in Gram-type solution of the discrete KP equation (\ref{HW-Gram}) can be rewritten as
\begin{eqnarray*}
  && m_{ij}=c_{ij}+\left( \frac{a_{1}-p_{i}}{a_{1}+q_{j}}\right) ^{-k_{1}}\left( \frac{%
a_{2}-p_{i}}{a_{2}+q_{j}}\right) ^{-k_{2}}\left( \frac{a_{3}-p_{i}}{a_{3}+q_{j}}\right) ^{-k_{3}} \\
  &&  \rightarrow c_{ij}+\left( -\frac{{p}_{i}}{
{q}_{j}}\right) ^{-k_{3}}  \left( \frac{a_{1}-p_{i}}{a_{1}+q_{j}}\right) ^{-k_{1}}\left( \frac{%
a_{2}-p_{i}}{a_{2}+q_{j}}\right) ^{-k_{2}} \left( \frac{1-a_{3}p^{-1}_{i}}{%
1+a_{3}q^{-1}_{j}}\right) ^{-k_{3}} \\
&& \rightarrow  c_{ij} + \left( -\frac{\widetilde{p}_{i}}{\widetilde{q}_{j}}\right) ^{-k_{1}}
  \left( -\frac{\widetilde{p}_{i}+a}{%
\widetilde{q}_{j}-a}\right) ^{-k_{3}} \left( \frac{1-b\widetilde{p}_{i}}{%
1+b\widetilde{q}_{j}}\right) ^{-k_{2}}  \left( \frac{1-a_{3}p^{-1}_{i}}{%
1+a_{3}q^{-1}_{j}}\right) ^{-k_{3}}
\end{eqnarray*}
by letting ${p}_i-a_1=\widetilde{p}_{i}$, ${q}_i+a_1=\widetilde{q}_{i}$ and $a_2-a_1=b^{-1}$, $a=a_{1}$.
Let $k_{3}=l$, $k_1=n$, $k_2=k$, then the discrete KP equation (\ref{H-M-1}) becomes
\begin{eqnarray}
&& (a_3-b^{-1}-a) \tau_{n}(k+1,k_2+1,k_3+1) \tau _{n+1}(k+1,k_{2},k_3) \nonumber \\
&&+(a-a_3 )\tau_{n}(k,k_2+1,k_3) \tau _{n+1}(k+1,k_{2},k_3+1)  \nonumber\\
&&+b^{-1} \tau _{n}(k+1,k_2,k_3+1)\tau _{n+1}(k,k_2+1,k_3) =0.
\label{full-discrete}
\end{eqnarray}
Applying Miwa transformation by taking $a_3 \to 0$ and $b \to 0,$
\begin{equation*}
x_{1} =kb\,,\ \  x_{2} =\frac{1}{2} kb^2, \ \ \cdots \ \  x_{n} =\frac{1}{n} kb^{n},
\end{equation*}
\begin{equation*}
x_{-1} =la_3\,,\ \   x_{-2} =\frac{1}{2} la^2_3, \ \ \cdots \ \  x_{-n} =\frac{1}{n} la^n_3,
\end{equation*}
we obtain an infinite number of bilinear equations as follows
\begin{eqnarray*}
&&\sum_{L,M} (a_3-b^{-1}-a ) b^{L} a_3^M  p_{L}\left(\frac{1}{2}\widetilde{D}_+ \right)
p_{M}\left( \frac{1}{2}\widetilde{D}_- \right) \tau _{n}(k+1)\cdot \tau _{n+1}(k)\\
&&+ \sum_{L,M} (a-a_3)  b^{-L} a_3^M  p_{L}
\left( \frac{1}{2}\widetilde{D}_+ \right) p_{M}\left( -\frac{1}{2}\widetilde{D}_- \right) \tau _{n}(k)\cdot \tau _{n+1}(k+1) \\
&&+ \sum_{L,M} b^{L-1} a_3^M  p_{L}
\left(- \frac{1}{2}\widetilde{D}_+\right) p_{M}\left( \frac{1}{2}\widetilde{D}_-\right) \tau _{n} (k+1)\cdot \tau _{n+1}(k)=0,
\end{eqnarray*}
where
\[
\widetilde{D}_+=\left(D_{x_{1}},\frac{1}{2}D_{x_{2}},\cdots ,\frac{1}{n}%
D_{x_{n}}\right)\,, \quad
\widetilde{D}_{-}=\left(D_{x_{-1}},\frac{1}{2}D_{x_{-2}},\cdots ,\frac{1}{n}%
D_{x_{-n}}\right)\,.
\]
At the order of $a_3^0b^0$, we have
\begin{eqnarray*}
&&(- p_{1}\left(\frac{1}{2}\widetilde{D}_+ \right)  -a) \tau _{n}(k+1)\cdot \tau _{n+1}(k)\\
&&+  a \tau _{n}(k)\cdot \tau _{n+1}(k+1) +  p_{1}
\left(- \frac{1}{2}\widetilde{D}_+\right)  \tau _{n} (k+1)\cdot \tau _{n+1}(k)=0,
\end{eqnarray*}
which gives Eq. (\ref{BL1}).

%\begin{eqnarray*}
%&&(D_1  + a) \tau _{n}(k+1)\cdot \tau _{n+1}(k)=a \tau _{n}(k)\cdot \tau _{n+1}(k+1)
%\end{eqnarray*}
At the order of $a_3^1b^0$, we have
\begin{eqnarray*}
&&(1- p_{1}\left(\frac{1}{2}\widetilde{D}_+ \right) p_{1}\left(\frac{1}{2}\widetilde{D}_- \right)  -a p_{1}\left(\frac{1}{2}\widetilde{D}_- \right) ) \tau _{n}(k+1)\cdot \tau _{n+1}(k)\\
&&+  (a p_{1}\left(-\frac{1}{2}\widetilde{D}_- \right)  -1) \tau _{n}(k)\cdot \tau _{n+1}(k+1)
+  p_{1}\left(- \frac{1}{2}\widetilde{D}_+\right)  p_{1}\left(\frac{1}{2}\widetilde{D}_- \right)   \tau _{n} (k+1)\cdot \tau _{n+1}(k)=0,
\end{eqnarray*}
which is equivalent to Eq. (\ref{BL2})
\end{proof}

In addition to the Gram-type determinant solution, the tau functions to above  two bilinear equations (\ref{BL1})--(\ref{BL2}) can also be expressed by the Casorati determinant, which is given by the following lemma.
\begin{lem}
	The  bilinear equations (\ref{BL1})--(\ref{BL2})
	have the following Casorati-type determinant solution
	\begin{equation}
		\tau_{n,k}=\left|\begin{array}{cccc}
	f_{n}^{(1)}(k) & f_{n+1}^{(1)}(k) & \cdots & f_{n+N-1}^{(1)}(k) \\
	f_{n}^{(2)}(k) & f_{n+1}^{(2)}(k) & \cdots & f_{n+N-1}^{(2)}(k) \\
	\cdots & \cdots & \cdots & \ldots \\
	f_{n}^{(N)}(k) & f_{n+1}^{(N)}(k) & \cdots & f_{n+N-1}^{(N)}(k)
	\end{array}\right|,
	\label{Casorati}
	\end{equation}
	where
	\begin{align}
	&f_{n}^{(i)}(k)=c_{i}p_i^n(p_i-a)^ke^{\xi_i}+d_iq_i^n(q_i-a)^ke^{\eta_i},\\
	&\xi_i=p_ix_1+\frac{1}{p_i-a}x_{-1}+\xi_{i0},\,\ \eta_i=q_ix_1+\frac{1}{q_i-a}x_{-1}+\eta_{i0}.
	\end{align}
	Here $c_{i},d_i,p_i,q_i,\xi_{i0},\eta_{i0}$ are arbitrary parameters.
\end{lem}
 The proof is given in Appendix A.

Now we proceed to an implementation of  period-2 reduction which will result in a negative KdV hierarchy. To this end, we impose the constraints
\begin{align}
q_i=-p_i\, \quad d_i = \mathrm{i} c_i\,, \label{red}
\end{align}
for Casorati-type solution (\ref{Casorati}) or
  \begin{align}
q_i=p_i\, \quad c_{ij}= \mathrm{i} \delta_{ij}\,, \label{blue}
\end{align}
for Gram-type solution (\ref{d2DTL-Gram}). As a result, we have
\begin{equation}
\tau_{n,k}\Bumpeq\tau_{n+2,k}, \ \ \tau_{n,k}\Bumpeq\bar{\tau}_{n+1,k}.
\end{equation}
Here $\Bumpeq$ means two $\tau$-functions are equivalent up to a constant multiple and $\bar{\tau}$ means the complex conjugate of $\tau$.

In addition, by taking variable transformations $y=x_1,\tau=\frac{2}{\kappa}x_{-1}$, $a=\frac{1}{2\kappa}$ and defining
\begin{align}
f=\tau_{00}, \ \ \tilde{f}=\tau_{10},\ \ g=\tau_{01},\ \ \tilde{g}=\tau_{11},\label{t}
\end{align}
we can obtain the following four bilinear equations.
\begin{align}
&D_{y} f \cdot \tilde{g}-\frac{1}{2 \kappa}(f \tilde{g}-\tilde{f} g)=0,\label{1}\\
&D_{y} \tilde{f} \cdot g-\frac{1}{2 \kappa}(\tilde{f} g-f \tilde{g})=0,\label{2}\\
&D_{\tau} D_{y} f \cdot \tilde{g}-\frac{1}{2 \kappa} D_{\tau} f \cdot \tilde{g}+\frac{1}{2 \kappa} D_{\tau} \tilde{f} \cdot g-\kappa(f \tilde{g}-\tilde{f} g)=0,\label{3} \\
&D_{\tau} D_{y} \tilde{f} \cdot g-\frac{1}{2 \kappa} D_{\tau} \tilde{f} \cdot g+\frac{1}{2 \kappa} D_{\tau} f \cdot \tilde{g}-\kappa(\tilde{f} g-f \tilde{g})=0, \label{4}
\end{align}
from Eqs. (\ref{BL1})--(\ref{BL2}). It should be pointed out that above four bilinear equations belong to the negative KdV hierarchy, which
derive the mCH equation as discovered in Ref. \cite{mats}. We summarize the result by the following lemma.
\begin{lem}
	The mCH equation (\ref{mch}) is derived from bilinear equations (\ref{1})--(\ref{4})
through the hodograph transformation
\begin{align}
&x=\frac{y}{\kappa}+\ln\frac{g\tilde{g}}{f\tilde{f}}+d,\label{hodo}\\
&t=\tau,
\end{align}
and dependent variables transformation
\begin{align}
u=\frac{1}{2\mathrm{i}\kappa}\left(\ln \frac{\tilde{f}\tilde{g}}{fg}\right)_\tau.\label{u}
\end{align}
\end{lem}
Although the proof is given in Ref. \cite{mats}, we outline the proof in Appendix \ref{pro-bl} for the purpose to mimic the process in constructing the semi-discretization of the mCH equation.

%Here, we provide an explicit representation of the N-soliton solution of mCH equation. Firstly, we consider the following lemma.

 Thus we are able to express the N-soliton solution of mCH equation (\ref{mch}) through the following theorem.
\begin{theo}
	The modified Camassa-Holm equation (\ref{mch}) admits a solution
	\begin{align*}
	& u=\frac{1}{2\mathrm{i}\kappa}\left(\ln \frac{\tilde{f}\tilde{g}}{fg}\right)_\tau,\\	
	&x=\frac{y}{\kappa}+\ln\frac{g\tilde{g}}{f\tilde{f}}+d,\quad  t=\tau\,,
	\end{align*}
	where $f,\tilde{f},g,\tilde{g}$ are the determinants given by (\ref{t}) and $\tau_{n,k}$ can be written either as a Gram-type determinant
	\begin{equation}
		\tau_{n,k}=\left| m^{n,k}_{ij}\right|=\left| \mathrm{i} \delta_{ij}+
		\left(-\frac{{p}_{i}}{{p}_{j}}\right)^{-n}\left(-\frac{{p}_{i}+\frac{1}{2\kappa}}{{p}_{j}-\frac{1}{2\kappa}}\right)^{-k} e^{\xi_i+\bar{\xi}_j} \right|,
		\label{Gram-sol}
\end{equation}
with
\begin{align*}
	\xi_i=p_iy+\frac{\kappa^2}{2\kappa p_i-1}\tau+\xi_{i0},\,\ \bar{\xi}_i=p_iy+\frac{\kappa^2}{2\kappa p_i+1}\tau+\bar{\xi}_{i0}\,,
\end{align*}	
	or a Casorati-type determinant
\begin{equation}
\tau_{n,k}=\left|\begin{array}{cccc}
f_{n}^{(1)}(k) & f_{n+1}^{(1)}(k) & \cdots & f_{n+N-1}^{(1)}(k) \\
f_{n}^{(2)}(k) & f_{n+1}^{(2)}(k) & \cdots & f_{n+N-1}^{(2)}(k) \\
\cdots & \cdots & \cdots & \ldots \\
f_{n}^{(N)}(k) & f_{n+1}^{(N)}(k) & \cdots & f_{n+N-1}^{(N)}(k)
\end{array}\right|,
\label{sol1}
\end{equation}
	with
	\begin{align}
	&f_{n}^{(i)}(k)=p_i^n(p_i-\frac{1}{2\kappa})^ke^{\xi_i}+\mathrm{i} (-p_i)^n(-p_i-\frac{1}{2\kappa})^ke^{\eta_i},\\
	&\xi_i=p_iy+\frac{\kappa^2}{2\kappa p_i-1}\tau+\xi_{i0},\,\ \eta_i=-p_iy-\frac{\kappa^2}{2\kappa p_i+1}\tau+\eta_{i0}\,.
	\end{align}
\end{theo}
\begin{remark}
The Casorati determinant solution (\ref{sol1}) is consistent with the solutions given in Ref. \cite{mats}. Same as many other soliton equations, the Gram-type determinant and the Casorati-type determinant solutions are basically equivalent to each other.
\end{remark}
\section{Integrable semi-discrete mCH  equation}\label{sec3}
In this section, we intend to construct the integrable spatial-discretization of the mCH equation. To this end, we first derive semi-discrete analogues of a set of bilinear equations (\ref{1})--(\ref{4})
in subsection \ref{sub-sec-3-1}. Then in subsection \ref{sub-sec-3-2}, we construct an integrable semi-discrete mCH equation.
\subsection{From discrete KP equation to the semi-discrete analogue of  (\ref{BL1}) and (\ref{BL2})}\label{sub-sec-3-1}
Firstly, we attempt to find semi-discrete analogues of (\ref{BL1}) and (\ref{BL2}). Amazingly, these semi-discrete equations can be obtained as an intermediate product in deriving bilinear equations of mCH equation (\ref{mch}). The result is outlined by the following lemma.
\begin{lem}
The discrete KP equation (\ref{H-M-1}) generates the following bilinear equations
\begin{align}
&\frac{1}{b}[\tau_{n,k+1}(l)\tau_{n+1,k}(l+1)-\tau_{n,k+1}(l+1)\tau_{n+1,k}(l)]\nonumber \\
&-a[\tau_{n,k}(l+1)\tau_{n+1,k+1}(l)-\tau_{n+1,k}(l)\tau_{n,k+1}(l+1)]=0,\label{dis1}\\
&\frac{1}{b}D_{x_{-1}}[\tau_{n,k+1}(l)\cdot\tau_{n+1,k}(l+1)-\tau_{n,k+1}(l+1) \cdot\tau_{n+1,k}(l)]\nonumber\\
&=(aD_{x_{-1}}-2)[\tau_{n+1,k+1}(l)\cdot \tau_{n,k}(l+1)-\tau_{n,k+1}(l+1)\cdot \tau_{n+1,k}(l)],\label{dis2}
\end{align}
which admit the determinant solution of Gram-type
\begin{equation}
		\tau_{n,k}(l)=\left| m^{n,k}_{ij}\right|=\left| c_{ij}+ \left(-\frac{{p}_{i}}{{q}_{j}}\right)^{-n}\left(-\frac{{p}_{i}+a}{{q}_{j}-a}\right)^{-k}  \left( \frac{1-b p_{i}}{%
1+b q_{j}}\right) ^{-l}  e^{\xi_i+\eta_j} \right|,
		\label{Gram1}
\end{equation}
where
\begin{align*}
	\xi_i=\frac{1}{p_i+a}x_{-1}+\xi_{i0},\,\ \eta_i=\frac{1}{q_i-a}x_{-1}+\eta_{i0}.
\end{align*}
\end{lem}
\begin{proof}
We apply the Miwa transformation to Eq. (\ref{full-discrete})  by taking $a_3 \rightarrow 0$ but leaving $b$ finite, then we have
\begin{align*}
&\sum\left(a_3-b^{-1}-a\right) a_3^{M} p_{M}\left(\frac{1}{2} \widetilde{D}_{-}\right) \tau_{n}(k+1, l+1) \cdot \tau_{n+1}(k, l) \\
&+\sum(a-a_3) d^{M} p_{M}\left(-\frac{1}{2} \widetilde{D}_{-}\right) \tau_{n}(k, l+1) \cdot \tau_{n+1}(k+1, l) \\
&+\sum b^{-1} a_3^{M} p_{M}\left(\frac{1}{2} \widetilde{D}_{-}\right) \tau_{n}(k+1, l) \cdot \tau_{n+1}(k, l+1)=0,
\end{align*}
by letting $k_{2}=l$.
At the order of $b^{0} a_3^{0}$, we have
\begin{align*}
\left(-\frac{1}{b}-a\right) \tau_{n}(k+1, l+1) \tau_{n+1}(k, l)+a \tau_{n}(k, l+1) \tau_{n+1}(k+1, l)+\frac{1}{b} \tau_{n}(k+1, l) \tau_{n+1}(k, l+1)=0,
\end{align*}
which is actually Eq. (\ref{dis1}).
%\begin{align*}
%\frac{1}{b}\left(\tau_{n}(k+1, l+1) \tau_{n+1}(k, l)-\tau_{n}(k+1, l) \tau_{n+1}(k, l+1)\right)=a\left(\tau_{n}(k, l+1) \tau_{n+1}(k+1, l)-\tau_{n}(k+1, l+1) \tau_{n+1}(k, l)\right)
%\end{align*}
At the order of $a_3^{1}$, we have
\begin{align*}
&\left(1-\left(b^{-1}+a\right) p_{1}\left(\frac{1}{2} \widetilde{D}_{-}\right)\right) \tau_{n}(k+1, l+1) \cdot \tau_{n+1}(k, l) \\
&+\left(a p_{1}\left(-\frac{1}{2} \widetilde{D}_{-}\right)-1\right) \tau_{n}(k, l+1) \cdot \tau_{n+1}(k+1, l)+\frac{1}{b} p_{1}\left(\frac{1}{2} \widetilde{D}_{-}\right) \tau_{n}(k+1, l) \cdot \tau_{n+1}(k, l+1)=0,
\end{align*}
which is nothing but Eq. (\ref{dis2}).
\end{proof}
The following lemma gives the Casorati-type determinant solution of bilinear equations (\ref{dis1}) and (\ref{dis2}).
\begin{lem}
	Bilinear equations (\ref{dis1}) and (\ref{dis2}) admit the following determinant solutions
	\begin{equation}
	\tau_{n,k}(l)=\left|\begin{array}{cccc}
	f_{n,k}^{(1)}(l) & f_{n+1,k}^{(1)}(l) & \cdots & f_{n+N-1,k}^{(1)}(l) \\
	f_{n,k}^{(2)}(l) & f_{n+1,k}^{(2)}(l) & \cdots & f_{n+N-1,k}^{(2)}(l) \\
	\cdots & \cdots & \cdots & \ldots \\
	f_{n,k}^{(N)}(l) & f_{n+1,k}^{(N)}(l) & \cdots & f_{n+N-1,k}^{(N)}(l)
	\end{array}\right|,
	\end{equation}
	where
	\begin{align}
	&f_{n,k}^{(i)}(l)=c_{i}p_i^n(p_i-a)^k(1-bp_i)^le^{\xi_i}+d_iq_i^n(q_i-a)^k(1-bq_i)^le^{\eta_i},\\
	&\xi_i=\frac{1}{p_i-a}x_{-1}+\xi_{i0},\,\ \eta_i=\frac{1}{q_i-a}x_{-1}+\eta_{i0}.
	\end{align}
	Here $c_{i},d_i,p_i,q_i,\xi_{i0},\eta_{i0}$ are the arbitrary parameters.
\end{lem}
The proof is presented in Appendix \ref{pro}. In order to realize the periodic-2 reduction in the discrete version, we  introduce an auxiliary index
$m$ so that
\begin{align}
&\frac{1}{b}\left[\tau_{n+1, k}(l,m) \tau_{n, k+1}(l+1,m)-\tau_{n, k+1}(l,m) \tau_{n+1, k}(l+1)\right] \nonumber \\
&=a\left(\tau_{n+1, k}(l,m) \tau_{n, k+1}(l+1,m)-\tau_{n, k}(l+1,m) \tau_{n+1, k+1}(l,m)\right), \label{d1} \\
&\frac{1}{b} D_{x_{-1}}\left[\tau_{n+1, k}(l,m) \cdot \tau_{n, k+1}(l+1,m)-\tau_{n+1, k}(l+1,m) \cdot \tau_{n, k+1}(l,m)\right] \nonumber \\
&=\left(a D_{x_{-1}}-2\right)\left(\tau_{n+1, k+1}(l,m) \cdot \tau_{n, k}(l+1,m)-\tau_{n, k+1}(l+1,m) \cdot \tau_{n+1, k}(l,m)\right), \label{d2}
\end{align}
and a set of parallel bilinear equations with parameter constant $c$ and a shift in $m$
\begin{align}
&\frac{1}{c}\left[\tau_{n+1, k}(l, m) \tau_{n, k+1}(l, m+1)-\tau_{n, k+1}(l, m) \tau_{n+1, k}(l, m+1)\right] \nonumber\\
&=a\left(\tau_{n+1, k}(l, m) \tau_{n, k+1}(l, m+1)-\tau_{n, k}(l, m+1) \tau_{n+1, k+1}(l, m)\right), \label{d3}\\
&\frac{1}{c} D_{x_{-1}}\left[\tau_{n+1, k}(l, m) \cdot \tau_{n, k+1}(l, m+1)-\tau_{n+1, k}(l, m+1) \cdot \tau_{n, k+1}(l, m)\right]\nonumber\\
&=\left(a D_{x_{-1}}-2\right)\left(\tau_{n+1, k+1}(l, m) \cdot \tau_{n, k}(l, m+1)-\tau_{n, k+1}(l, m+1) \cdot \tau_{n+1, k}(l, m)\right)\,. \label{d4}
\end{align}
The solution to bilinear equations (\ref{d1})--(\ref{d4}) can be expressed as the Casorati-type determinant
\begin{equation}
\tau_{n,k}(l,m)=\left|\begin{array}{cccc}
f_{n,k}^{(1)}(l,m) & f_{n+1,k}^{(1)}(l,m) & \cdots & f_{n+N-1,k}^{(1)}(l,m) \\
f_{n,k}^{(2)}(l,m) & f_{n+1,k}^{(2)}(l,m) & \cdots & f_{n+N-1,k}^{(2)}(l,m) \\
\cdots & \cdots & \cdots & \ldots \\
f_{n,k}^{(N)}(l,m) & f_{n+1,k}^{(N)}(l,m) & \cdots & f_{n+N-1,k}^{(N)}(l,m)
\end{array}\right|\, ,
\end{equation}
where
\begin{align}
&f_{n,k}^{(i)}(l,m)=c_{i}p_i^n(p_i-a)^k(1-bp_i)^l(1-cp_i)^me^{\xi_i}+d_iq_i^n(q_i-a)^k(1-bq_i)^l(1-cq_i)^me^{\eta_i},\\
&\xi_i=\frac{1}{p_i-a}x_{-1}+\xi_{i0},\,\ \eta_i=\frac{1}{q_i-a}x_{-1}+\eta_{i0}\,,
\end{align}
or Gram-type determinant
\begin{equation}
		\tau_{n,k}(l)=\left| c_{ij}+ \left(-\frac{{p}_{i}}{{q}_{j}}\right)^{-n}\left(-\frac{{p}_{i}+a}{{q}_{j}-a}\right)^{-k}  \left( \frac{1-b p_{i}}{%
1+b q_{j}}\right) ^{-l} \left( \frac{1-c p_{i}}{%
1+c q_{j}}\right) ^{-m}  e^{\xi_i+\eta_j} \right|\,,
		\label{Gram2}
\end{equation}
where
\begin{align}
&\xi_i=\frac{1}{p_i-a}x_{-1}+\xi_{i0},\,\ \eta_i=\frac{1}{q_i+a}x_{-1}+\eta_{i0}\,.
\end{align}
%Since $l$ and $m$ are equivalent, the proof of solution above is the same as in Proposition 1.
To realize period-2 reduction in both the continuous and discrete cases simultaneously, we impose constraints
\begin{align*}
b=-c\,, \quad q_{i}=-p_{i}\,, i=1, \cdots, N,
\end{align*}
for Casorati-type solution
or
\begin{align*}
b=-c\,, \quad q_{i}=p_{i}\,, i=1, \cdots, N,
\end{align*}
for Gram-type solution.  Under these conditions, we have
\begin{align}
&\tau_{n,k}(l,m) \Bumpeq\tau_{n+2,k}(l,m), \ \
\tau_{n,k} (l,m) \Bumpeq\bar{\tau}_{n+1,k} (l,m), \ \
\tau_{n,k} (l+1,m+1) \Bumpeq {\tau}_{n,k} (l,m)\,.
\end{align}
Under this reduction, we can drop index $m$ and define
\begin{align}
f_{l}=\tau_{00 l}, \quad \tilde{f}_{l}=\tau_{10 l}, \quad g_{l}=\tau_{01 l}, \quad \tilde{g}_{l}=\tau_{11 l}.
\end{align}

Same as previous section, $\tilde{f}_l$ and $\tilde{g}_l$ can be set to be complex conjugate of $f_l$ and $g_l$, respectively. Furthermore,  let $x_{-1}=\frac{\kappa}{2} \tau$, $a=\frac{1}{2\kappa}$ and take $n=1, k=0$, we have
\begin{align}
&\frac{1}{b}\left(f_{l} \tilde{g}_{l+1}-f_{l+1} \tilde{g}_{l}\right)=\frac{1}{2 \kappa}\left(f_{l} \tilde{g}_{l+1}-g_{l} \tilde{f}_{l+1}\right),\label{dis1'} \\
&-\frac{1}{b}\left(f_{l} \tilde{g}_{l-1}-f_{l-1} \tilde{g}_{l}\right)=\frac{1}{2 \kappa}\left(f_{l} \tilde{g}_{l-1}-g_{l} \tilde{f}_{l-1}\right),\label{dis2'} \\
&\frac{1}{b} D_{\tau }\left(f_{l} \cdot \tilde{g}_{l+1}-f_{l+1} \cdot \tilde{g}_{l}\right)=\left(\frac{1}{2 \kappa } D_{x_{-1}}- \kappa \right)\left(g_{l} \cdot \tilde{f}_{l+1}-\tilde{g}_{l+1} \cdot f_{l}\right), \label{dis3'}\\
&-\frac{1}{b} D_{\tau}\left(f_{l} \cdot \tilde{g}_{l-1}-f_{l-1} \cdot \tilde{g}_{l}\right)=\left(\frac{1}{2\kappa} D_{\tau}- \kappa\right)\left(g_{l} \cdot \tilde{f}_{l-1}-\tilde{g}_{l-1} \cdot f_{l}\right). \label{dis4'}
\end{align}
For $n=0, k=0$, we have
\begin{align}
&\frac{1}{b}\left(\tilde{f}_{l} g_{l+1}-\tilde{f}_{l+1} g_{l}\right)=\frac{1}{2 \kappa}\left(\tilde{f}_{l} g_{l+1}-\tilde{g}_{l} f_{l+1}\right), \label{dis1''}\\
&-\frac{1}{b}\left(\tilde{f}_{l} g_{l-1}-\tilde{f}_{l-1} g_{l}\right)=\frac{1}{2 \kappa}\left(\tilde{f}_{l} g_{l-1}-\tilde{g}_{l} f_{l-1}\right),\label{dis2''} \\
&\frac{1}{b} D_{\tau}\left(\tilde{f}_{l} \cdot g_{l+1}-\tilde{f}_{l+1} \cdot g_{l}\right)=\left(\frac{1}{2 \kappa} D_{\tau}-\kappa \right)\left(\tilde{g}_{l} \cdot f_{l+1}-g_{l+1} \cdot \tilde{f}_{l}\right),\label{dis3''} \\
&-\frac{1}{b} D_{\tau}\left(\tilde{f}_{l} \cdot g_{l-1}-f_{l-1} \cdot \tilde{g}_{l}\right)=\left(\frac{1}{2 \kappa} D_{\tau}-\kappa \right)\left(\tilde{g}_{l} \cdot f_{l-1}-g_{l-1} \cdot \tilde{f}_{l}\right).\label{dis4''}
\end{align}
\begin{remark}
The bilinear equations (\ref{dis1'})--(\ref{dis4''}) are discrete analogues of the set of bilinear equations (\ref{1})--(\ref{4}), which constitute a set of crucial bilinear equations for us to construct integrable semi-discretization of the mCH equation (\ref{mch}).
\end{remark}
\subsection{Integrable semi-discrete mCH eqaution}\label{sub-sec-3-2}
We propose an integrable  semi-discrete mCH equation, the main result of the present paper, by the following theorem.
\begin{theo}
An integrable semi-discrete analogue of the mCH equation (\ref{mch}) is derived as
\begin{align}
&  \frac{d \delta_l}   {d t}   =2 m_l (\delta u_l),  \label{semi-mCH1} \\
& m_l =  \frac 12 (u_l+u_{l+1}) -  \Gamma_l  r_l \tilde{r}_l   (\delta^2 u_l)-  \Gamma_l  r_l (\delta \tilde{r}_{l}) (\delta \tilde{u}_{l}),   \label{semi-mCH2}
\end{align}
from Eqs. (\ref{dis1'})--(\ref{dis4''}), through a dependent variable transformation
\begin{align}
u_l=\frac{1}{2\mathrm{i} \kappa} \left(\ln\frac{\tilde{f}_l\tilde{g}_l}{f_lg_l}\right)_{\tau},
\end{align}
and a discrete hodograph transformation
\begin{align}
\delta_l=\frac{x_{l+1}-x_l} {b}=\frac{1}{r_{l}}=\frac{1}{\kappa} \cos \phi_{l}\,, \quad t=\tau\,.
\end{align}
Here $\cos \phi_{l}$ is defined from
\begin{align}
\frac{g_{l} \tilde{f}_{l+1}+g_{l+1} \tilde{f}_{l}}{f_{l} \tilde{g}_{l+1}+f_{l+1} \tilde{g}_{l}}=e^{\mathrm{i} \phi_{l}}, \quad \frac{\tilde{g}_{l} f_{l+1}+\tilde{g}_{l+1} f_{l}}{\tilde{f}_{l} g_{l+1}+\tilde{f}_{l+1} g_{l}}=e^{-\mathrm{i} \phi_{l}},
\end{align}
and other variables are defined by
\begin{align}
& m_{l}=\kappa\tan \phi_{l}, \quad \Delta \left( \ln\frac{{f}_l}{\tilde{g}_l}\right) = \frac{2}{b} \frac{f_{l} \tilde{g}_{l+1}-f_{l+1} \tilde{g}_{l}}{f_{l} \tilde{g}_{l+1}+f_{l+1} \tilde{g}_{l}}, \\
&\delta u_l=\frac{1}{2\mathrm{i} \kappa}\left(  \Delta \left( \ln\frac{\tilde{f}_l\tilde{g}_l}{f_lg_l}\right)\right)_{\tau}\,,\quad
\delta^2 u_l=\frac{\delta u_{l+1}-\delta u_l}{b},\\
& \tilde{r}_l  = \frac{\kappa}{\cos\frac{\phi_{l+1}+\phi_l}{2}}\,, \ \  \delta \tilde{r}_{l}  = \frac{2\kappa \tan\frac{\phi_{l+1}+\phi_l}{2} }{\cos\frac{\phi_{l+1}+\phi_l}{2}}  \frac{ \tan\frac{\phi_{l+1}-\phi_l}{2}}{b},\\
& \tilde{\phi}_l=\mathrm{i}\ln \frac{f_l\tilde{g}_l}{\tilde{f}_lg_l},\label{phi_l}\, \quad \quad \delta \tilde{u}_{l}  = - \frac{1}{2 \kappa^2} \left( \sin\frac{\phi_{l+1}+\phi_l}{2}\right)_\tau,
\end{align}
and
\begin{align}
\Gamma_l=\frac{\left({1}-\frac{b^2}{8 \kappa^2}(1-\cos\phi_l)\right)(\tilde{\phi}_{l+1}-\tilde{\phi}_{l})_{\tau}}{(\phi_{l+1}-\phi_l)_\tau}.
\end{align}
%Here the difference operator $\Delta$ satisfies $\Delta F_n=F_{n+1}-F_n.$
 \end{theo}
 Prior to the proof of the theorem, let us show that the proposed semi-discrete mCH equation (\ref{semi-mCH1})--(\ref{semi-mCH2}) converges to the mCH equation
 (\ref{mch}) in the continuum limit $b \to 0$. \\
 Recall that
 \begin{align*}
u= \frac{1}{2\mathrm{i} \kappa}  \left(\ln\frac{\tilde{f}\tilde{g}}{fg}\right)_{\tau}, \ \ x= \frac{y}{\kappa} + \ln \frac{g \tilde{g}}{f \tilde{f}}, \ \ \phi=\mathrm{i} \ln \frac{f \tilde{g}}{\tilde{f} g}.
\end{align*}
 It is obvious
 that as $b \to 0$, we have
 \begin{eqnarray*}
m_l \to m=\kappa \tan \phi, \ \ \phi_l, \tilde{\phi}_l \to \phi,  \ \ r_l, \tilde{r}_l \to r= \frac{\kappa}{\cos \phi}=\frac{\partial x}{\partial y}, \ \
\delta \tilde{r}_{l} \to  r_y\,.
\end{eqnarray*}
It follows that, as $b \to 0$, $\Gamma_l \to 1$,  $\tilde{g}_{l+1} \to \tilde{g}_{l}-b \tilde{g}_{l, y}$, $f_{l+1} \to f_{l}-b f_{l, y}$, and thus we obtain
\begin{align*}
&\frac{2}{b} \frac{f_{l} \tilde{g}_{l+1}-f_{l+1} \tilde{g}_{l}}{f_{l} \tilde{g}_{l+1}+f_{l+1} \tilde{g}_{l}} \to  \frac{2(f_{l, y} \tilde{g}_{l}-f_{l} \tilde{g}_{l, y})}{2 f_{l} \tilde{g}_{l}-b f_{l} \tilde{g}_{l+1, y}-b f_{l, y} \tilde{g}_{l}} \to (\ln f / \tilde{g})_{y}\,.
\end{align*}
Therefore
\begin{align*}
\delta u_l \to \frac{1}{2\mathrm{i} \kappa}  \left(\ln\frac{\tilde{f}\tilde{g}}{fg}\right)_{y \tau} \to u_y, \ \  \delta \tilde{u}_{l}\to u_y, \ \ \delta^2 u_l \to u_{yy}.
\end{align*}
Consequently, Eq. (\ref{semi-mCH2}) converges to
\begin{eqnarray}
&& m=u - r^2 u_{yy}-rr_y u_y = u-u_{xx}\,,
\end{eqnarray}
while Eq. (\ref{semi-mCH1}) converges to
\begin{eqnarray*}
&& \left( \frac{1}{r} \right)_\tau - 2 m u_y =0\,,
\end{eqnarray*}
or
\begin{eqnarray}
&&r _\tau + 2 m r^2 u_y =0\,,
\label{Amch}
\end{eqnarray}
where $r=\sqrt{m^2+\kappa^2}$.
On the other hand, Eq. (\ref{semi-mCH1}) is equivalent to
\begin{eqnarray*}
&& \frac{\partial^2 x} {\partial y \partial \tau} = 2 m u_y   = 2 (u- r(ru_y)_y) u_y = (u^2-(r^2 u^2_y))_y\,,
\end{eqnarray*}
or \begin{eqnarray*}
&& \frac{\partial x} {\partial \tau} =u^2-u^2_x\,,
\end{eqnarray*}
which implies
\begin{eqnarray*}
&&\partial_\tau = \partial_t+   (u^2-u^2_x) \partial_x\,.
\end{eqnarray*}
As a result, Eq. (\ref{Amch})  leads to
\begin{eqnarray*}
&& (\partial_t+ (u^2-u^2_x) \partial_x) r + r(u^2-u^2_x)_x =0,
\label{Amch1}
\end{eqnarray*}
which is exactly the mCH equation (\ref{mch}). In what follows, we give the detailed proof of the theorem.
\begin{proof}
Shifting $l$ to $l+1$, Eqs. (\ref{dis2'}) and (\ref{dis2''}) can be rewriten as
\begin{align}
&\frac{1}{b}\left(f_{l} \tilde{g}_{l+1}-f_{l+1} \tilde{g}_{l}\right)=\frac{1}{2 \kappa}\left(f_{l+1} \tilde{g}_{l}-g_{l+1} \tilde{f}_{l}\right), \label{dis2a''}\\
&\frac{1}{b}\left(\tilde{f}_{l} g_{l+1}-\tilde{f}_{l+1} g_{l}\right)=\frac{1}{2 \kappa}\left(\tilde{f}_{l+1} g_{l}-\tilde{g}_{l+1} f_{l}\right). \label{dis2b''}
\end{align}
Notice that we obtain a relation
\begin{align}
f_{l} \tilde{g}_{l+1}-g_{l} \tilde{f}_{l+1}=f_{l+1} \tilde{g}_{l}-g_{l+1} \tilde{f}_{l},\label{trans}
\end{align} from Eqs. (\ref{dis1'}) and (\ref{dis2a''}).
By adding (\ref{dis1'})  to (\ref{dis2a''}), and (\ref{dis1''})  to (\ref{dis2b''}), we have
\begin{align}
&\frac{2}{b} \frac{f_{l} \tilde{g}_{l+1}-f_{l+1} \tilde{g}_{l}}{f_{l} \tilde{g}_{l+1}+f_{l+1} \tilde{g}_{l}}=\frac{1}{2 \kappa}\left(1-\frac{g_{l} \tilde{f}_{l+1}+g_{l+1} \tilde{f}_{l}}{f_{l} \tilde{g}_{l+1}+f_{l+1} \tilde{g}_{l}}\right), \label{1c} \\
&\frac{2}{b} \frac{\tilde{f}_{l} {g}_{l+1}-\tilde{f}_{l+1} {g}_{l}}{\tilde{f}_{l} {g}_{l+1}+\tilde{f}_{l+1}{g}_{l}}=\frac{1}{2 \kappa}\left(1-\frac{\tilde{g}_{l} f_{l+1}+\tilde{g}_{l+1} f_{l}}{\tilde{f}_{l} g_{l+1}+\tilde{f}_{l+1} g_{l}}\right).\label{2c}
\end{align}
These two are discrete analogues of Eqs. (\ref{1'}) and (\ref{2'}).
Adding and subtracting Eqs. (\ref{1c}) and (\ref{2c}) give
\begin{align}
&\frac{2}{b}\left(\frac{f_{l} \tilde{g}_{l+1}-f_{l+1} \tilde{g}_{l}}{f_{l} \tilde{g}_{l+1}+f_{l+1} \tilde{g}_{l}}+\frac{\tilde{f}_{l} g_{l+1}-\tilde{f}_{l+1} g_{l}}{\tilde{f}_{l} g_{l+1}+\tilde{f}_{l+1} g_{l}}\right)=\frac{1}{\kappa}\left(1-\frac{1}{2}\left[\frac{g_{l} \tilde{f}_{l+1}+g_{l+1} \tilde{f}_{l}}{f_{l} \tilde{g}_{l+1}+f_{l+1} \tilde{g}_{l}}+\frac{\tilde{g}_{l} f_{l+1}+\tilde{g}_{l+1} f_{l}}{\tilde{f}_{l} g_{l+1}+\tilde{f}_{l+1} g_{l}}\right]\right), \\
&\frac{2}{b}\left(\frac{f_{l} \tilde{g}_{l+1}-f_{l+1} \tilde{g}_{l}}{f_{l} \tilde{g}_{l+1}+f_{l+1} \tilde{g}_{l}}-\frac{\tilde{f}_{l} g_{l+1}-\tilde{f}_{l+1} g_{l}}{\tilde{f}_{l} g_{l+1}+\tilde{f}_{l+1} g_{l}}\right)=\frac{1}{2 \kappa}\left(\frac{\tilde{g}_{l} f_{l+1}+\tilde{g}_{l+1} f_{l}}{\tilde{f}_{l} g_{l+1}+\tilde{f}_{l+1} g_{l}}-\frac{g_{l} \tilde{f}_{l+1}+g_{l+1} \tilde{f}_{l}}{f_{l} \tilde{g}_{l+1}+f_{l+1} \tilde{g}_{l}}\right),
\end{align}
which  are discrete analogues of Eqs. (\ref{3'}) and (\ref{4'}), respectively.
They can be further abbreviated as
\begin{align}
\Delta\left(\ln f_{l} \tilde{f}_{l} / g_{l} \tilde{g}_{l}\right) &=\frac{1}{\kappa}\left(1-\cos \phi_{l}\right), \\
\Delta\left(\ln  \tilde{f}_{l} \tilde{g}_{l}/ f_{l} g_{l}\right) &=\frac{1}{\mathrm{i} \kappa} \sin \phi_{l}\,. \label{dis-2}
\end{align}
Referring to the definition of  $\delta u_l$ and differentiating (\ref{dis-2}) with respect to $\tau$, we have
\begin{align*}
\delta u_l= -\frac{1}{2\kappa^2}  (\sin \phi_{l})_{\tau}=  \frac{1}{2} \frac{\cos\phi_l}{\kappa \sin\phi_l}\left(\frac{\cos\phi_l}{\kappa}\right)_\tau=  \frac{1}{2}  \frac{1}{m_l}\left(\frac{1}{r_l}\right)_\tau,
\end{align*}
or equivalently,
\begin{align}
(r_l)_\tau+2 r_l^2m_l (\delta u_l)=0.
\label{mCH-discrete1}
\end{align}
This is a discrete analogue of Eq. (\ref{eq1'}).  Since $\delta_l = r^{-1}_l$, Eq. (\ref{mCH-discrete1}) can be rewritten as
\begin{align}
 \frac{d \delta_l}{d\tau}  =2 m_l (\delta u_l),
%\label{mCH-discrete1}
\end{align}
which constitutes the first equation of the semi-discrete mCH equation. Now let us proceed to deducing the second equation of the semi-discrete mCH equation corresponding to Eq. (\ref{m}). Adding Eqs. (\ref{dis3'}) and (\ref{dis4'}) leads to
\begin{align*}
&\frac{2}{b} D_{\tau}\left(f_{l} \cdot \tilde{g}_{l+1}-f_{l+1} \cdot \tilde{g}_{l}\right)=\left(\frac{1}{2 \kappa} D_{\tau}+\kappa\right)\left(f_{l} \cdot \tilde{g}_{l+1}+f_{l+1} \cdot \tilde{g}_{l}-\tilde{f}_{l+1} \cdot g_{l}-\tilde{f}_{l} \cdot g_{l+1}\right).
\end{align*}
By using the relation (\ref{5}), one can obtain
\begin{align*}
&\frac{1}{b}\left(\left(f_{l} \tilde{g}_{l+1}+f_{l+1} \tilde{g}_{l}\right)\left(\ln \frac{f_{l} \tilde{g}_{l}}{f_{l+1} \tilde{g}_{l+1}}\right)_{\tau}+\left(f_{l} \tilde{g}_{l+1}-f_{l+1} \tilde{g}_{l}\right)\left(\ln \frac{f_{l} f_{l+1}}{\tilde{g}_{l} \tilde{g}_{l+1}}\right)_{\tau}\right) \\
&=\frac{1}{4 \kappa}\left(\left(f_{l} \tilde{g}_{l+1}-f_{l+1} \tilde{g}_{l}\right)\left(\ln \frac{f_{l} \tilde{g}_{l}}{f_{l+1} \tilde{g}_{l+1}}\right)_{\tau}+\left(f_{l} \tilde{g}_{l+1}+f_{l+1} \tilde{g}_{l}\right)\left(\ln \frac{f_{l} f_{l+1}}{\tilde{g}_{l} \tilde{g}_{l+1}}\right)_{\tau}\right. \\
&\left.+\left(\tilde{f}_{l} g_{l+1}-\tilde{f}_{l+1} g_{l}\right)\left(\ln \frac{\tilde{f}_{l} g_{l}}{\tilde{f}_{l+1} g_{l+1}}\right)_{\tau}-\left(\tilde{f}_{l} g_{l+1}+\tilde{f}_{l+1} g_{l}\right)\left(\ln \frac{\tilde{f}_{l} \tilde{f}_{l+1}}{g_{l} g_{l+1}}\right)_{\tau}\right) \\
&+\kappa\left(f_{l} \tilde{g}_{l+1}+f_{l+1} \tilde{g}_{l}-\tilde{f}_{l+1} g_{l}-\tilde{f}_{l} g_{l+1}\right).
\end{align*}

Dividing both sides by $f_{l} \tilde{g}_{l+1}+f_{l+1} \tilde{g}_{l}$, we have
\begin{align}
&\frac{1}{b}\left(\ln \frac{f_{l} \tilde{g}_{l}}{f_{l+1} \tilde{g}_{l+1}}\right)_{\tau}+\frac{f_{l} \tilde{g}_{l+1}-f_{l+1} \tilde{g}_{l}}{f_{l} \tilde{g}_{l+1}+f_{l+1} \tilde{g}_{l}}\left(\ln \frac{f_{l} f_{l+1}}{\tilde{g}_{l} \tilde{g}_{l+1}}\right)_{\tau} \nonumber \\
&=\frac{1}{4 \kappa}\left(\frac{f_{l} \tilde{g}_{l+1}-f_{l+1} \tilde{g}_{l}}{f_{l} \tilde{g}_{l+1}+f_{l+1} \tilde{g}_{l}}\left(\ln \frac{f_{l} \tilde{g}_{l}}{f_{l+1} \tilde{g}_{l+1}}\right)_{\tau}+\left(\ln \frac{f_{l} f_{l+1}}{\tilde{g}_{l} \tilde{g}_{l+1}}\right)_{\tau}\right. \nonumber \\
&\left.+\frac{\tilde{f}_{l} g_{l+1}-\tilde{f}_{l+1} g_{l}}{f_{l} \tilde{g}_{l+1}+f_{l+1} \tilde{g}_{l}}\left(\ln \frac{\tilde{f}_{l} g_{l}}{\tilde{f}_{l+1} g_{l+1}}\right)_{\tau}-\frac{\tilde{f}_{l} g_{l+1}+\tilde{f}_{l+1} g_{l}}{f_{l} \tilde{g}_{l+1}+f_{l+1} \tilde{g}_{l}}\left(\ln \frac{\tilde{f}_{l} \tilde{f}_{l+1}}{g_{l} g_{l+1}}\right)_{\tau}\right) \nonumber \\
&+\kappa\left(1-\frac{\tilde{f}_{l+1} g_{l}+\tilde{f}_{l} g_{l+1}}{f_{l} \tilde{g}_{l+1}+f_{l+1} \tilde{g}_{l}}\right).
\label{inta}
\end{align}
By a substitution of Eqs. (\ref{1c}) and (\ref{2c}) into Eq. (\ref{inta}), we obtain
\begin{align}
\frac{1}{b}\left(\ln \frac{f_{l} \tilde{g}_{l}}{f_{l+1} \tilde{g}_{l+1}}\right)_{\tau}=&\frac{1}{4 \kappa}\left(\frac{\tilde{f}_{l} g_{l+1}-\tilde{f}_{l+1} g_{l}}{f_{l} \tilde{g}_{l+1}+f_{l+1} \tilde{g}_{l}}\left(\ln \frac{\tilde{f}_{l} f_{l+1}g_l\tilde{g}_{l+1}}{f_l\tilde{f}_{l+1} \tilde{g}_l g_{l+1}}\right)_{\tau}-\frac{\tilde{f}_{l} g_{l+1}+\tilde{f}_{l+1} g_{l}}{f_{l} \tilde{g}_{l+1}+f_{l+1} \tilde{g}_{l}}\left(\ln \frac{\tilde{f}_{l}\tilde{f}_{l+1} \tilde{g}_{l}\tilde{g}_{l+1}}{f_{l} f_{l+1}g_{l}g_{l+1}}\right)_{\tau}\right) \nonumber\\
&+\kappa\left(1-\frac{\tilde{f}_{l+1} g_{l}+\tilde{f}_{l} g_{l+1}}{f_{l} \tilde{g}_{l+1}+f_{l+1} \tilde{g}_{l}}\right).
\label{1d}
\end{align}
This is a discrete analogue of Eq. (\ref{3''}).
%\begin{align*}
%(\ln f \tilde{g})_{\tau y}+\frac{1}{2 \kappa} \frac{\tilde{f} g}{f \tilde{g}}\left(\ln \frac{\tilde{f} \tilde{g}}{f g}\right)_{\tau}-\kappa\left(1-\frac{\tilde{f} g}{f \tilde{g}}\right)=0.
%\end{align*}
Taking the complex conjugate of (\ref{1d}), we have
\begin{align}
\frac{1}{b}\left(\ln \frac{\tilde{f}_{l} {g}_{l}}{\tilde{f}_{l+1} {g}_{l+1}}\right)_{\tau}=&\frac{1}{4 \kappa}\left(\frac{{f}_{l} \tilde{g}_{l+1}-{f}_{l+1} \tilde{g}_{l}}{\tilde{f}_{l} {g}_{l+1}+\tilde{f}_{l+1} {g}_{l}}\left(\ln \frac{{f}_{l} \tilde{f}_{l+1}\tilde{g}_l g_{l+1}}{\tilde{f}_l{f}_{l+1} {g}_l \tilde{g}_{l+1}}\right)_{\tau}-\frac{f_{l} \tilde{g}_{l+1}+f_{l+1} \tilde{g}_{l}}{\tilde{f}_{l} g_{l+1}+\tilde{f}_{l+1} g_{l}}\left(\ln \frac{f_{l} f_{l+1}g_{l}g_{l+1}}{\tilde{f}_{l}\tilde{f}_{l+1} \tilde{g}_{l}\tilde{g}_{l+1}}\right)_{\tau}\right) \nonumber\\
&+\kappa\left(1-\frac{f_{l} \tilde{g}_{l+1}+f_{l+1} \tilde{g}_{l}}{\tilde{f}_{l+1} g_{l}+\tilde{f}_{l} g_{l+1}}\right)\,,
\label{2d}
\end{align}
which is a discrete analogue of Eq. (\ref{4''}).
%which is a discrete analogue of
%\begin{align*}
%(\ln \tilde{f} g)_{\tau y}+\frac{1}{2 \kappa} \frac{f \tilde{g}}{\tilde{f} g}\left(\ln \frac{f g}{\tilde{f} \tilde{g}}\right)_{\tau}-\kappa\left(1-\frac{f \tilde{g}}{\tilde{f} g}\right)=0
%\end{align*}
Subtracting above two equations (\ref{1d}) and (\ref{2d}), one yields
\begin{align*}
&\frac{1}{b}\left(\ln \frac{f_{l} \tilde{g}_{l} \tilde{f}_{l+1} g_{l+1}}{\tilde{f}_{l} g_{l} f_{l+1} \tilde{g}_{l+1}}\right)_{\tau} \\
&=\frac{1}{4 \kappa}\left(\frac{\tilde{f}_{l} g_{l+1}-\tilde{f}_{l+1} g_{l}}{f_{l} \tilde{g}_{l+1}+f_{l+1} \tilde{g}_{l}}+\frac{{f}_{l} \tilde{g}_{l+1}-{f}_{l+1} \tilde{g}_{l}}{\tilde{f}_{l} {g}_{l+1}+\tilde{f}_{l+1} {g}_{l}}\right)\left(\ln \frac{\tilde{f}_{l} f_{l+1}g_l\tilde{g}_{l+1}}{f_l\tilde{f}_{l+1} \tilde{g}_l g_{l+1}}\right)_{\tau} \\
&\quad-\frac{1}{4 \kappa}(\frac{\tilde{f}_{l} g_{l+1}+\tilde{f}_{l+1} g_{l}}{f_{l} \tilde{g}_{l+1}+f_{l+1} \tilde{g}_{l}}+\frac{f_{l} \tilde{g}_{l+1}+f_{l+1} \tilde{g}_{l}}{\tilde{f}_{l} g_{l+1}+\tilde{f}_{l+1} g_{l}})\left(\ln \frac{\tilde{f}_{l}\tilde{f}_{l+1} \tilde{g}_{l}\tilde{g}_{l+1}}{f_{l} f_{l+1}g_{l}g_{l+1}}\right)_{\tau} \\
&\quad+\kappa\left(\frac{f_{l} \tilde{g}_{l+1}+f_{l+1} \tilde{g}_{l}}{\tilde{f}_{l+1} g_{l}+\tilde{f}_{l} g_{l+1}}-\frac{\tilde{f}_{l+1} g_{l}+\tilde{f}_{l} g_{l+1}}{f_{l} \tilde{g}_{l+1}+f_{l+1} \tilde{g}_{l}}\right)\,,
\end{align*}
which can be rewritten as
\begin{align}
&\left(\frac{1}{b}+A\right)\left(\ln \frac{f_{l} \tilde{g}_{l} \tilde{f}_{l+1} g_{l+1}}{\tilde{f}_{l} g_{l} f_{l+1} \tilde{g}_{l+1}}\right)_{\tau}\nonumber\\
&+\frac{1}{4 \kappa}(\frac{\tilde{f}_{l} g_{l+1}+\tilde{f}_{l+1} g_{l}}{f_{l} \tilde{g}_{l+1}+f_{l+1} \tilde{g}_{l}}+\frac{f_{l} \tilde{g}_{l+1}+f_{l+1} \tilde{g}_{l}}{\tilde{f}_{l} g_{l+1}+\tilde{f}_{l+1} g_{l}})\left(\ln \frac{\tilde{f}_{l}\tilde{f}_{l+1} \tilde{g}_{l}\tilde{g}_{l+1}}{f_{l} f_{l+1}g_{l}g_{l+1}}\right)_{\tau}\nonumber\\
&-\kappa\left(\frac{f_{l} \tilde{g}_{l+1}+f_{l+1} \tilde{g}_{l}}{\tilde{f}_{l+1} g_{l}+\tilde{f}_{l} g_{l+1}}-\frac{\tilde{f}_{l+1} g_{l}+\tilde{f}_{l} g_{l+1}}{f_{l} \tilde{g}_{l+1}+f_{l+1} \tilde{g}_{l}}\right)=0,\label{dis-c}
\end{align}
where
\begin{align*}
A=\frac{1}{4 \kappa}\left(\frac{\tilde{f}_{l} g_{l+1}-\tilde{f}_{l+1} g_{l}}{f_{l} \tilde{g}_{l+1}+f_{l+1} \tilde{g}_{l}}+\frac{{f}_{l} \tilde{g}_{l+1}-{f}_{l+1} \tilde{g}_{l}}{\tilde{f}_{l} {g}_{l+1}+\tilde{f}_{l+1} {g}_{l}}\right)=\frac{b}{8 \kappa^2}(-1+\cos\phi_l)\,.
\end{align*}
which approaches zero as $b \to 0$.
Eq. (\ref{dis-c})  is a discrete analogue of Eq. (\ref{c}).

Using the definition of $\tilde{\phi}_l$, Eq. (\ref{dis-c}) can be rewritten as
\begin{align}
\left(\frac{1}{b}-\frac{b}{8 \kappa^2}(1-\cos\phi_l)\right)(\tilde{\phi}_{l+1}-\tilde{\phi}_{l})_{\tau}+\cos\phi_l(u_l+u_{l+1})-2\kappa \sin\phi_l=0.
\end{align}
On the other hand,
\begin{align*}
b \delta^2 u_l&= \frac{1}{2\mathrm{i}\kappa}\left(\Delta\left(\ln \frac{\tilde{f}_{l+1} \tilde{g}_{l+1} } { f_{l+1} g_{l+1}}\right) -\Delta\left(\ln \frac{\tilde{f}_{l} \tilde{g}_{l} } { f_{l} g_{l}}\right)\right)_\tau\\
&=-\frac{1}{2\kappa^2}(\sin\phi_{l+1}-\sin\phi_l)_\tau=-\frac{1}{2\kappa^2}\left(2\cos\frac{\phi_{l+1}+\phi_l}{2} \sin\frac{\phi_{l+1}-\phi_l}{2}\right)_\tau\\
&=-\frac{1}{2\kappa^2}\left(-\sin\frac{\phi_{l+1}+\phi_l}{2} \sin\frac{\phi_{l+1}-\phi_l}{2}(\phi_{l+1}+\phi_l)_\tau+\cos\frac{\phi_{l+1}+\phi_l}{2} \cos\frac{\phi_{l+1}-\phi_l}{2} (\phi_{l+1}-\phi_l)_\tau\right).
\end{align*}
Referring to the definition of $\Gamma_l$, which converges to $1$ as $b\to 0$, we have
\begin{align}
&- 2 \kappa^2  \delta^2 u_l + \sin\frac{\phi_{l+1}+\phi_l}{2} \sin\frac{\phi_{l+1}-\phi_l}{2}(\phi_{l+1}+\phi_l)_\tau \nonumber \\
&=- \Gamma^{-1}_l \cos\frac{\phi_{l+1}+\phi_l}{2} \cos\frac{\phi_{l+1}-\phi_l}{2} \left( \cos\phi_l(u_l+u_{l+1})-2\kappa \sin\phi_l\right)\,.
\end{align}
As a result, we obtain
\begin{align}
m_l&=   \frac 12 (u_l+u_{l+1})  -\Gamma_l \frac{  \kappa^2  \delta^2 u_l}{ \cos\phi_l \cos\frac{\phi_{l+1}+\phi_l}{2}}
+ \Gamma_l \tan\frac{\phi_{l+1}+\phi_l}{2} \tan\frac{\phi_{l+1}-\phi_l}{2} \frac{(\phi_{l+1}+\phi_l)_\tau}{2 b \cos\phi_l } \nonumber \\
&=  \frac 12 (u_l+u_{l+1}) -  \Gamma_l  r_l \tilde{r}_l   (\delta^2 u_l) -   \Gamma_l  r_l (\delta \tilde{r}_{l}) (\delta \tilde{u}_{l})
\end{align}
by using the definitions of  $m_l$, $\tilde{r}_l$,  $\delta \tilde{r}_{l}$  and $\delta \tilde{u}_{l}$.  The proof is complete.
\end{proof}
\subsection{One- and two-soliton solutions}
\subsubsection{One soliton solutions}
The $\tau$-functions for the one-soliton solution of the mCH equation (\ref{mch}) are
\begin{align*}
f\propto 1+\mathrm{i}e^{-\zeta},\quad g\propto 1+\mathrm{i}\frac{1+2\kappa p}{1-2\kappa p}e^{-\zeta},
\end{align*}
with
\begin{align*}
\zeta=2p(y-\frac{2\kappa^3}{1-4\kappa^2p^2}t)+\zeta_{0},
\end{align*}
where we set $p=p_1$ for simplicity. Thus we can obtain the one-soliton solution in a parametric form
\begin{align}
u&=\frac{-8\kappa^2p}{(2\kappa p+1)(2\kappa p-1)}\frac{(2\kappa p-1)e^{3\zeta}-(2\kappa p+1)e^{\zeta}}{(2\kappa p-1)^2e^{4\zeta}+(4\kappa^2p^2+1)e^{2\zeta}+(2\kappa p+1)^2}\nonumber\\
&=\frac{-8\kappa^2p}{(1-4\kappa^2 p^2)^{\frac{3}{2}}}\frac{\cosh(\tilde{\zeta})}{\cosh(2\tilde{\zeta})+\frac{1+4\kappa^2 p^2}{1-4\kappa^2 p^2}},\\
X&=x-ct-x_0=\frac{\tilde{\zeta}}{2\kappa p}+\ln\frac{1-2\kappa p \tanh(\tilde{\zeta})}{1+2\kappa p\tanh(\tilde{\zeta})},
\end{align}
with
\begin{align*}
\tilde{\zeta}=\zeta+\varphi,\,e^\varphi=\sqrt{\frac{1-2\kappa p}{1+2\kappa p}},\, c=\frac{2\kappa^2}{1-4\kappa^2 p^2},
\end{align*}
where $c$ stands for the velocity of the soliton and $x_0=y_0/\kappa$ has been chosen such that $X=0$ when $\tilde{\zeta}=0$. One can see that the one-soliton solution above is equivalent to the solution in Ref. \cite{mats} by setting $k=2p$.

For the semi-discrete mCH equation, the $ \tau$-functions are
\begin{align*}
f_l\propto 1+\mathrm{i}\left(\frac{1+bp}{1-bp}\right)^le^{-\theta},\quad  g_l\propto 1-\mathrm{i}\left(\frac{1+bp}{1-bp}\right)^l\frac{2\kappa p+ 1}{2\kappa p-1}e^{-\theta},
\end{align*}
with $\theta=-\frac{4\kappa^3p}{1-4\kappa^2p^2}\tau+\theta_{0}$. By taking $b=-0.1$, Figure \ref{1-soliton-fig} depicts a one-soliton solution to the semi-discrete mCH equation while comparing with the one to the mCH equation with different values of $p$. When $p=0.3$,  $u$ is single-valued and owns one peak since $X_y>0$ as shown in  Figure \ref{mch-1}. Figure \ref{mch-2} illustrates the symmetric singular soliton in Ref. \cite{mats} which is three-valued with  two spikes for $p=0.45$, while $u$ becomes antisymmetric for $p=0.6$ (see Figure \ref{mch-3}).
\begin{figure}[H]
	\centering
	\subfigure[]
	{
		\label{mch-1}
		\includegraphics[width=1.7in]{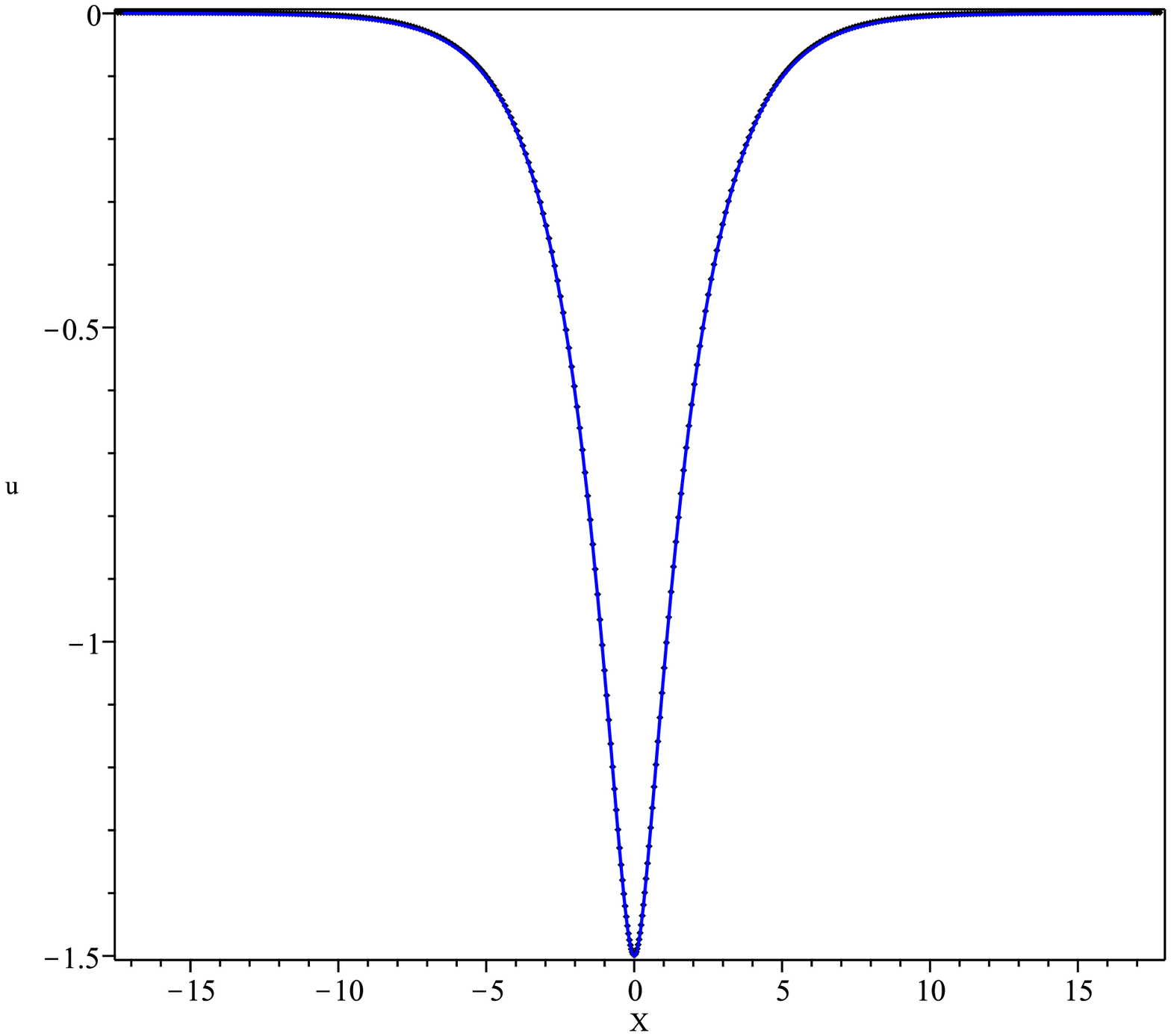}
	}
	\hspace*{3em}
	\subfigure[]
	{\label{mch-2}
		\includegraphics[width=1.7in]{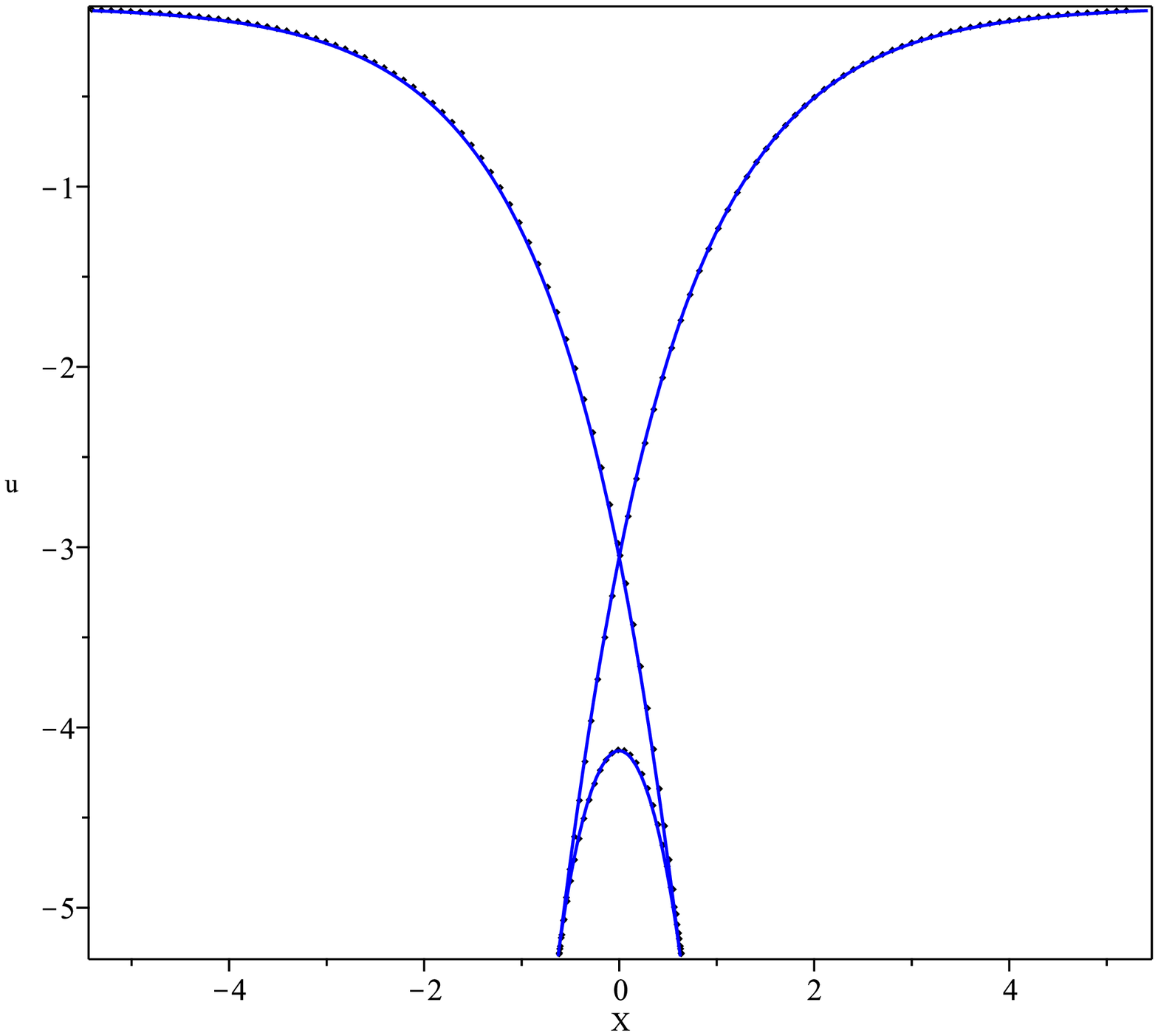}
	}
    \hspace*{3em}
	\subfigure[]
    {\label{mch-3}
	\includegraphics[width=1.7in]{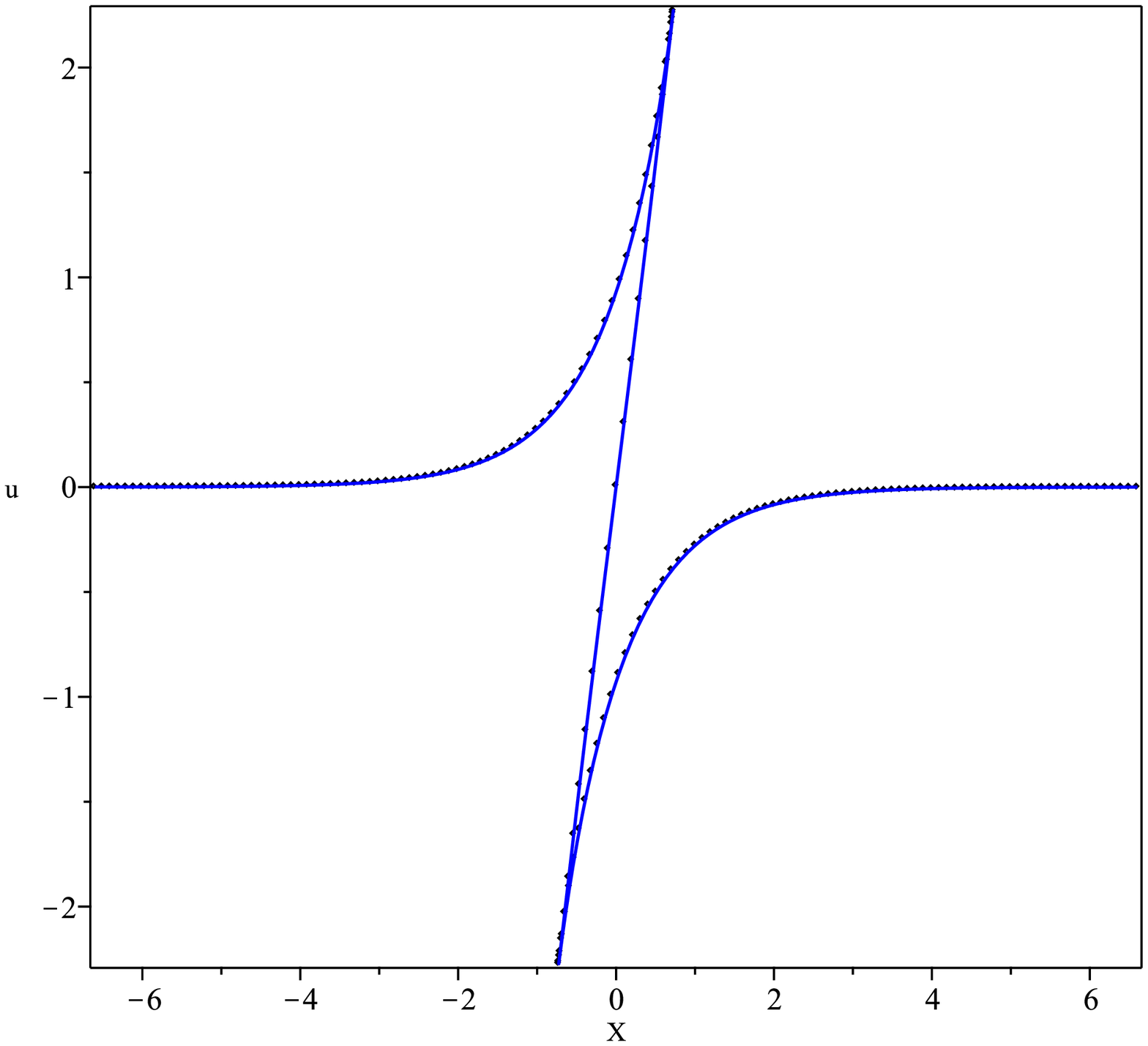}
    }
	\caption{Comparison between the one-soliton solution of the mCH equation and the semi-discrete mCH equation with $\kappa=1$ at $t=0$; solid line: mCH equation, dot: semi-discrete mCH equation. (a) $p=0.3$, (b) $p=0.45$, (c) $p=0.6$.}\label{1-soliton-fig}
\end{figure}
\subsubsection{Two soliton solutions}
The $\tau $-functions for two soliton solutions of the mCH equation (\ref{mch}) are
\begin{align}
&f\propto 1+\mathrm{i}\frac{p_1+p_2}{p_2-p_1}e^{-\zeta_1}-\mathrm{i}\frac{p_1+p_2}{p_2-p_1}e^{-\zeta_2}+e^{-\zeta_1-\zeta_2},\\
&g\propto1+\mathrm{i}a_1\frac{p_1+p_2}{p_2-p_1}e^{-\zeta_1}-\mathrm{i}a_2\frac{p_1+p_2}{p_2-p_1}e^{-\zeta_2}+a_1a_2e^{-\zeta_1-\zeta_2},
\end{align}
with
\begin{align}
&\zeta_i=2p_i(y-\frac{2\kappa^3}{1-4\kappa^2p_i^2}t)+\zeta_{i0},\,i=1,2,\\
&a_i=\frac{1+2\kappa p_i}{1-2 \kappa p_i},\,i=1,2.
\end{align}
For the semi-discrete mCH equation, the $\tau$-functions are
\begin{align}
&f_l\propto 1+\mathrm{i}\frac{p_1+p_2}{p_2-p_1}z_1^le^{-\theta_1}-
\mathrm{i}\frac{p_1+p_2}{p_2-p_1}z_2^le^{-\theta_2}
+(z_1z_2)^le^{-\theta_1-\theta_2},\\
&g_l\propto 1+\mathrm{i}a_1\frac{p_1+p_2}{p_2-p_1}z_1^l e^{-\theta_1}-\mathrm{i}a_2\frac{p_1+p_2}{p_2-p_1}z_2^l e^{-\theta_2}+a_1a_2(z_1z_2)^le^{-\theta_1-\theta_2},
\end{align}
with $z_i= \frac{1+bp_i}{1-bp_i}$,  $\theta_i=-\frac{4\kappa^3p_i}{1-4\kappa^2p_i^2}t+\theta_{i0},\,i=1,2$. As pointed out in Ref. \cite{mats}, if we set $p_2=p_1^*$,  two-soliton solution becomes a breather solution. Here $^*$ means complex conjugate. Figure \ref{2-soliton-fig} displays such a breather solution with
$p_1=0.1+0.2\mathrm{i}$ and $b=-0.1$. It  can be found that the  solution of the semi-discrete mCH equation agrees with that of the mCH equation very well.
\begin{figure}[H]
	\centering
	\subfigure[]
	{
		\label{mch-4}
		\includegraphics[width=2.2in]{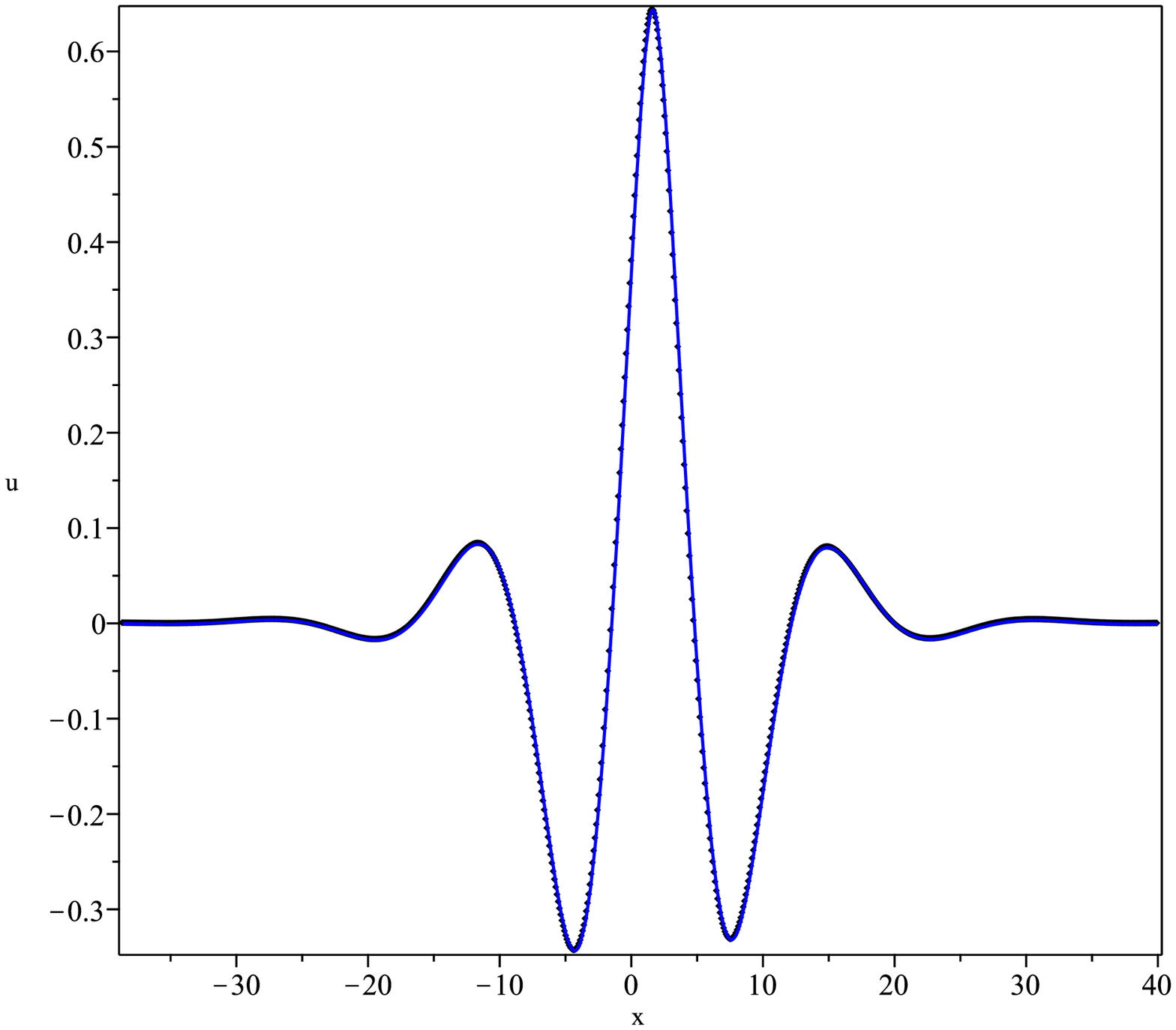}
	}
	\hspace*{3em}
	\subfigure[]
	{\label{mch-5}
		\includegraphics[width=2.2in]{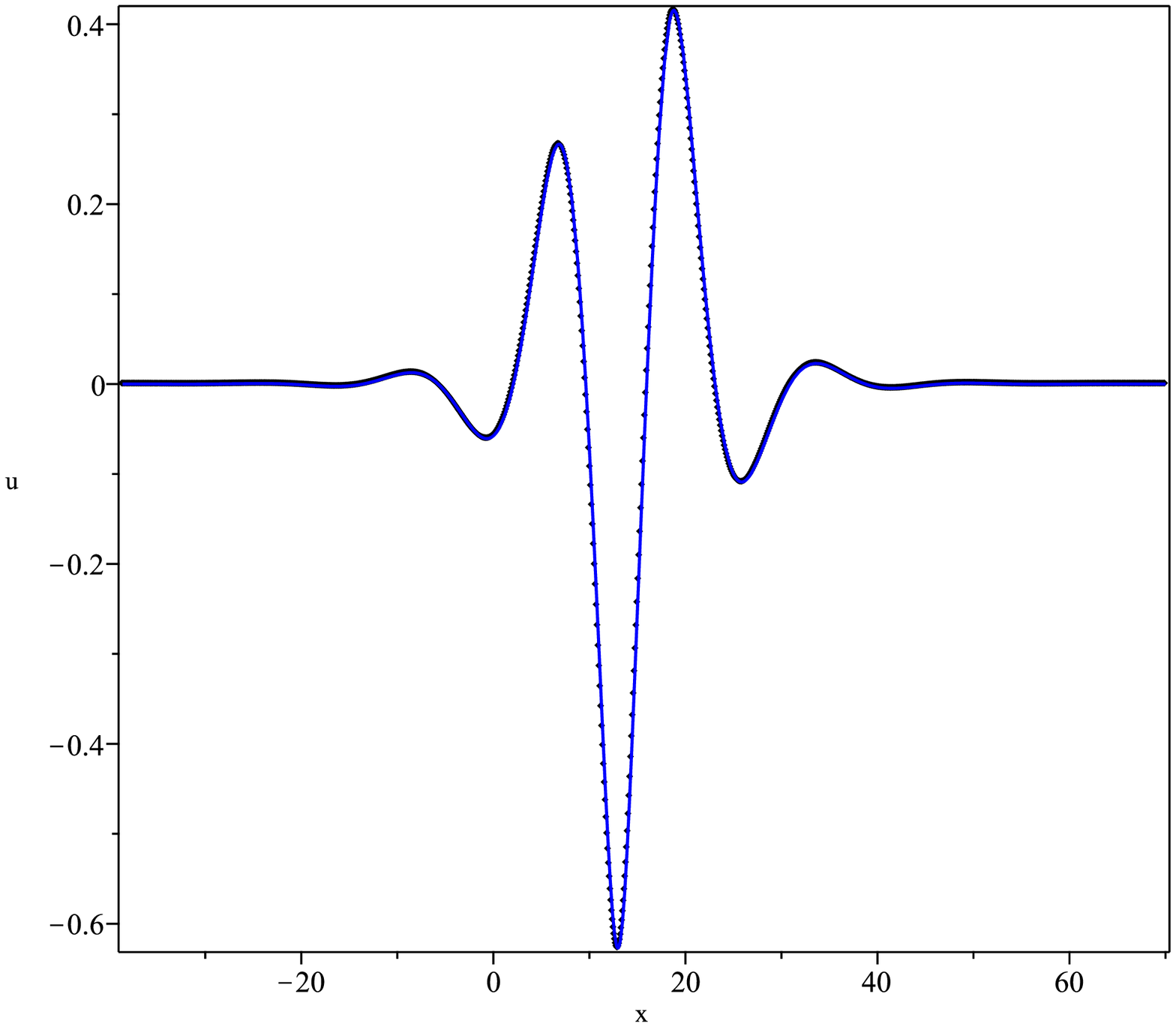}
	}
	\caption{Comparison between the two-soliton solution of the mCH and the semi-discrete mCH equation for $p_1=0.1+0.2\mathrm{i},\,p_2=0.1-0.2\mathrm{i}$ with $\kappa=1$; solid line: mCH equation, dot: semi-discrete mCH equation. (a) $t=0$, (b) $t=10$.}\label{2-soliton-fig}
\end{figure}
\section{Conclusion}\label{sec4}
In this paper, starting from the discrete KP equation, we have constructed an integrable semi-discrete analogue of the modified Camassa-Holm equation with cubic nonlinearity and linear dispersion through Hirota's bilinear approach and a series of reductions.  The $N$-soliton solution for the semi-discrete mCH equation is provided and approved in  both the Gram-type and the Casorati-type determinant forms. It is confirmed again that the discrete KP equation is one of the fundamental equations for integrable systems following the work by Hirota, Ohta, Tsujimoto and Nimmo etc. However, several questions arise naturally. First, what is the Lax pair of this integrable semi-discrete mCH equation?  It becomes possible for further investigation of this semi-discrete mCH equation such as the inverse scattering transform, the symmetry and the conservation law if the Lax pair is known. Actually, there is a general question related: how could we obtain the Lax pair for the reduced discrete integrable systems based on the Lax pair of discrete KP equation.  Second, can we construct the connection between the discrete KP equation and some known integrable systems such as the two-component CH equation\upcite{YoujinLMP}, the complex short pulse equation\upcite{Feng15}, the Fokas-Lennels equation\upcite{FLa,FLb} and the massive Thirring model equation\upcite{Thirring, MikhailovMT}? Third:, once we know the these connections, can we construct the integrable discretizations of these soliton equations with import physical applications? The last question: can we construct other types of soliton solutions such as rogue wave solutions for both the discrete and continuous equations by using their connection to the discrete KP equation? We expect to have more subsequent  papers  in the near future by the authors and other researchers in answering above questions.  

\section*{Acknowledgement}
GY's work is supported by National Natural Science Foundation of China (Grant No. 11871336). BF's work is
partially supported by National Science Foundation (NSF) under Grant No. DMS-1715991 and U.S. Department of Defense (DoD), Air Force for Scientific
Research (AFOSR) under grant No. W911NF2010276.

\appendix
\section{Proof of Lemma 2.2}\label{pro-cont}

\begin{proof}
From the expression of $f_{n}^{(i)}(k)$, we note that
\begin{align}
f_{n+1,k}^{(i)}=f_{n,k+1}^{(i)}+af_{n,k}^{(i)},\,\partial_{x_1}f_{n,k}^{(i)}=f_{n+1,k}^{(i)},\,\partial_{x_{-1}}f_{n,k}^{(i)}=f_{n,k-1}^{(i)}.
\end{align}	
For convenience, we introduce a notation
\begin{equation}
|0_k,1_k,\cdots,N-1_k|=\left|\begin{array}{cccc}
f_{n}^{(1)}(k) & f_{n+1}^{(1)}(k) & \cdots & f_{n+N-1}^{(1)}(k) \\
f_{n}^{(2)}(k) & f_{n+1}^{(2)}(k) & \cdots & f_{n+N-1}^{(2)}(k) \\
\cdots & \cdots & \cdots & \ldots \\
f_{n}^{(N)}(k) & f_{n+1}^{(N)}(k) & \cdots & f_{n+N-1}^{(N)}(k)
\end{array}\right|.
\end{equation}
We can verify the following relations
\begin{align}
\tau_{n+1,k}&=|1_k,2_k,\cdots,N_k|=|1_k,1_{k+1},\cdots,N-1_{k+1}|,\\
\tau_{n,k}&=|0_k,0_{k+1},\cdots,N-2_{k+1}|,\\
\partial_{x_1}\tau_{n,k+1}&=|0_{k+1},1_{k+1},\cdots,N-2_{k+1},N_{k+1}|\\
\partial_{x_1}\tau_{n+1,k}&=a|1_k,1_{k+1},\cdots,N-1_{k+1}|+|1_k,1_{k+1},\cdots,N-2_{k+1},N_{k+1}|
\end{align}
It is noticed that
\begin{align}
-\tau_{n,k+1} \partial_{x_1}\tau_{n+1,k}+a\tau_{n,k+1}\tau_{n+1,k}&=-|0_{k+1},\cdots,N-1_{k+1}||1_k,1_{k+1},\cdots,N-2_{k+1},N_{k+1}|\\
a\tau_{n+1,k+1}\tau_{n,k}
&=|1_{k+1},\cdots,N_{k+1}|||1_k,0_{k+1},\cdots,N-2_{k+1}|
\end{align}
Therefore, the Pl{\"u}cker relation for determinant
\begin{align}
\begin{split}
&|0_{k+1},\cdots,N-1_{k+1}||1_k,1_{k+1},\cdots,N-2_{k+1},N_{k+1}|\\
&-|0_{k+1},1_{k+1},\cdots,N-2_{k+1},N_{k+1}||1_k,1_{k+1},\cdots,N-1_{k+1}|\\
&+|1_{k+1},\cdots,N_{k+1}|||1_k,0_{k+1},\cdots,N-2_{k+1}|=0.
\end{split}
\end{align}
gives
\begin{align}
\partial_{x_1}\tau_{n,k+1}\tau_{n+1,k}-\tau_{n,k+1} \partial_{x_1}\tau_{n+1,k}+a\tau_{n,k+1}\tau_{n+1,k}=a\tau_{n+1,k+1}\tau_{n,k}
\end{align}
which is the bilinear equation (\ref{BL1}).

Next, we prove the bilinear equation (\ref{BL2}) by rewriting the l.h.s. into two parts
\begin{align*}
D_{x_{-1}}(D_{x_1}+a)\tau_{n, k+1}\cdot \tau_{n+1, k}
&=D_{x_{-1}}(\partial_{x_1}\tau_{n, k+1}\cdot\tau_{n+1,k}-\tau_{n, k+1}\cdot\partial_{x_1}\tau_{n+1,k}+a\tau_{n, k+1}\cdot \tau_{n+1, k})\\
&=D_{x_{-1}}(\partial_{x_1}\tau_{n, k+1}\cdot\tau_{n+1,k})+D_{x_{-1}}(-\tau_{n, k+1}\cdot\partial_{x_1}\tau_{n+1,k}+a\tau_{n, k+1}\cdot \tau_{n+1, k})\\
&=I_1+I_2,
\end{align*}
where
\begin{align*}
I_1&=D_{x_{-1}}(\partial_{x_1}\tau_{n, k+1}\cdot\tau_{n+1,k})\\
&=D_{x_{-1}}(|0_{k+1},1_{k+1},\cdots,N-2_{k+1},N_{k+1}|\cdot |1_k,1_{k+1},\cdots,N-1_{k+1}|)
\end{align*}
\begin{align*}
I_2&=D_{x_{-1}}(-\tau_{n, k+1}\cdot\partial_{x_1}\tau_{n+1,k}+a\tau_{n, k+1}\cdot \tau_{n+1, k})\\
&=D_{x_{-1}}(-|0_{k+1},\cdots,N-1_{k+1}|\cdot |1_k,1_{k+1},\cdots,N-2_{k+1},N_{k+1}|).
\end{align*}
Since
\begin{align*}
&\partial_{x_{-1}}|0_{k+1},1_{k+1},\cdots,N-2_{k+1},N_{k+1}|=
\sum_{i=0,i\neq N-1}^{N}|0_{k+1},1_{k+1},\cdots,i_k,\cdots,N-2_{k+1},N_{k+1}|,\\
&\partial_{x_{-1}}|1_k,1_{k+1},\cdots,N-1_{k+1}|=|1_{k-1},1_{k+1},\cdots,N-1_{k+1}|,\\
&\partial_{x_{-1}}|0_{k+1},\cdots,N-1_{k+1}|=\sum_{i=0,}^{N-1}|0_{k+1},1_{k+1},\cdots,i_k,\cdots,N-1_{k+1}|,\\
&\partial_{x_{-1}}|1_k,1_{k+1},\cdots,N-2_{k+1},N_{k+1}|=|1_{k-1},1_{k+1},\cdots,N-2_{k+1},N_{k+1}|+|1_{k},1_{k+1},\cdots,N-2_{k+1},N_{k}|,\\
&\partial_{x_{-1}}|1_{k+1},\cdots,N_{k+1}|=\sum_{i=1}^{N}|1_{k+1},\cdots,i_k,\cdots,N_{k+1}|,\\
&\partial_{x_{-1}}|0_k,0_{k+1},\cdots,N-2_{k+1}|=|0_{k-1},0_{k+1},\cdots,N-2_{k+1}|,
\end{align*}
we can verify the following Pl{\"u}cker relations for determinants
\begin{align*}
&|0_{k+1},1_{k+1},\cdots,i_k,\cdots,N-2_{k+1},N_{k+1}||1_k,1_{k+1},\cdots,N-1_{k+1}|\\
&-|0_{k+1},\cdots,i_k,\cdots,N-1_{k+1}||1_k,1_{k+1},\cdots,N-2_{k+1},N_{k+1}|\\
&=-a^{i-1}|1_k,1_{k+1},\cdots,0_{k+1},\cdots,N-2_{k+1},N_{k+1}||1_k,1_{k+1},\cdots,i_{k+1},\cdots,N-1_{k+1}|\\
&+a^{i-1}|1_k,1_{k+1},\cdots,0_{k+1},\cdots,N-1_{k+1}||1_k,1_{k+1},\cdots,i_{k+1},\cdots,N-2_{k+1},N_{k+1}|\\
&=a|1_{k+1},\cdots,i_k,\cdots,N_{k+1}||0_k,0_{k+1},\cdots,N-2_{k+1}|,i=1,\cdots,N-2\\
\\
&-|0_{k+1},1_{k+1},\cdots,N-2_{k+1},N_{k+1}||1_{k-1},1_{k+1},\cdots,N-1_{k+1}|\\
&+|0_{k+1},1_{k+1},\cdots,N-2_{k+1},N-1_{k+1}||1_{k-1},1_{k+1},\cdots,N-2_{k+1},N_{k+1}|\\
&=-|1_{k+1},\cdots,N_{k+1}||1_{k-1},0_{k+1},1_{k+1},\cdots,N-2_{k+1}|\\
&=-a\tau_{n+1,k+1}\partial_{x_{-1}}\tau_{n,k}-\tau_{n+1,k+1}\tau_{n,k}.
\end{align*}
Now we just need to verify
\begin{align*}
&(|0_{k},1_{k+1},\cdots,N-2_{k+1},N_{k+1}|+|0_{k+1},1_{k+1},\cdots,N-2_{k+1},N_{k}|)|1_k,1_{k+1},\cdots,N-1_{k+1}|\\
&-(|0_k,1_{k+1},\cdots,N-1_{k+1}|+|0_{k+1},1_{k+1},\cdots,N-1_{k}|)|1_k,1_{k+1},\cdots,N-2_{k+1},N_{k+1}|\\
&+|0_{k+1},\cdots,N-1_{k+1}||1_k,1_{k+1},\cdots,N-2_{k+1},N_{k}|\\
&-2|1_k,1_{k+1},\cdots,N-1_{k+1}||0_{k+1},1_{k+1},\cdots,N-1_{k+1}|\\
&=a(|1_{k+1},\cdots,N-2_{k+1},N-1_{k},N_{k+1}|+|1_{k+1},\cdots,N-1_{k+1},N_{k}|)|0_k,0_{k+1},\cdots,N-2_{k+1}|\\
&-|1_{k+1},\cdots,N_{k+1}||0_k,0_{k+1},\cdots,N-2_{k+1}|.
\end{align*}
which can be transformed into
\begin{align*}
&a|1_k,1_{k+1},\cdots,N-1_{k+1}||0_{k+1},1_{k+1},\cdots,N-2_{k+1},N-1_{k}|\\
&-|0_{k+1},1_{k+1},\cdots,N-1_{k}||1_k,1_{k+1},\cdots,N-2_{k+1},N_{k+1}|\\
&=a(|1_{k+1},\cdots,N-2_{k+1},N-1_{k},N_{k+1}|+|1_{k+1},\cdots,N-1_{k+1},N_{k}|)|0_k,0_{k+1},\cdots,N-2_{k+1}|
\end{align*}
by using
\begin{align*}
&|0_{k},1_{k+1},\cdots,N-2_{k+1},N_{k+1}||1_k,1_{k+1},\cdots,N-1_{k+1}|\\
&-|0_k,1_{k+1},\cdots,N-1_{k+1}||1_k,1_{k+1},\cdots,N-2_{k+1},N_{k+1}|\\
&=-|1_{k+1},\cdots,N_{k+1}||0_k,0_{k+1},1_{k+1},\cdots,N-2_{k+1}|,\\
\\
&|0_{k+1},1_{k+1},\cdots,N-2_{k+1},N_{k}||1_k,1_{k+1},\cdots,N-1_{k+1}|\\
&-|1_k,1_{k+1},\cdots,N-1_{k+1}||0_{k+1},1_{k+1},\cdots,N-1_{k+1}|\\
&=a|1_k,1_{k+1},\cdots,N-1_{k+1}||0_{k+1},1_{k+1},\cdots,N-2_{k+1},N-1_{k}|\\
\\
&|0_{k+1},\cdots,N-1_{k+1}||1_k,1_{k+1},\cdots,N-2_{k+1},N_{k}|\\
&-|1_k,1_{k+1},\cdots,N-1_{k+1}||0_{k+1},1_{k+1},\cdots,N-1_{k+1}|\\
&=a|1_k,1_{k+1},\cdots,N-2_{k+1},N-1_{k}||0_{k+1},\cdots,N-1_{k+1}|=0\,.
\end{align*}
Since
\begin{align*}
&|1_k,1_{k+1},\cdots,N-1_{k+1}||0_{k+1},1_{k+1},\cdots,N-2_{k+1},N-1_{k}|\\
&=|0_{k+1},1_{k+1},\cdots,N-1_{k+1}||1_k,1_{k+1},\cdots,N-2_{k+1},N-1_{k}|\\
&+|1_{k+1},\cdots,N-1_{k+1},N-1_{k}||1_k,0_{k+1},1_{k+1},\cdots,N-2_{k+1}|\\
&=|1_{k+1},\cdots,N-1_{k+1},N_{k}||0_k,0_{k+1},1_{k+1},\cdots,N-2_{k+1}|\\
\\
&-|0_{k+1},1_{k+1},\cdots,N-1_{k}||1_k,1_{k+1},\cdots,N-2_{k+1},N_{k+1}|\\
&=-|0_{k+1},1_{k+1},\cdots,N-2_{k+1},N_{k+1}||1_k,1_{k+1},\cdots,N-1_{k}|\\
&+|1_{k+1},\cdots,N-1_k,N_{k+1}||1_k,0_{k+1},\cdots,N-2_{k+1}|\\
&=a|1_{k+1},\cdots,N-1_k,N_{k+1}||0_k,0_{k+1},\cdots,N-2_{k+1}|,
\end{align*}
Eq. (\ref{BL2}) is verified.
\end{proof}
\section{Proof of Lemma 2.3}\label{pro-bl}
\begin{proof}
	Firstly, the bilinear equations (\ref{1})--(\ref{4}) can be recast into
	\begin{align}
	&(\ln\frac{f}{\tilde{g}})_y=\frac{1}{2\kappa}(1-\frac{\tilde{f}g}{f\tilde{g}}),\label{1'}\\
	&(\ln\frac{\tilde{f}}{g})_y=\frac{1}{2\kappa}(1-\frac{f\tilde{g}}{\tilde{f}g}),\label{2'}\\
	&(\ln f\tilde{g})_{\tau y}+(\ln\frac{f}{\tilde{g}})_\tau (\ln\frac{f}{\tilde{g}})_y-\frac{1}{2\kappa}(\ln\frac{f}{\tilde{g}})_\tau+\frac{1}{2\kappa}\frac{\tilde{f}g}{f\tilde{g}}(\ln\frac{\tilde{f}}{g})_\tau-\kappa(1-\frac{\tilde{f}g}{f\tilde{g}})=0,\label{3'}\\
	&(\ln \tilde{f}{g})_{\tau y}+(\ln\frac{\tilde{f}}{g})_\tau (\ln\frac{\tilde{f}}{g})_y-\frac{1}{2\kappa}(\ln\frac{\tilde{f}}{g})_\tau+\frac{1}{2\kappa}\frac{f\tilde{g}}{\tilde{f}g}(\ln\frac{f}{\tilde{g}})_\tau-\kappa(1-\frac{f\tilde{g}}{\tilde{f}g})=0,\label{4'}
	\end{align}
	by using
	\begin{align}
	&\frac{D_yf\cdot \tilde{g}}{f\tilde{g}}=(\ln\frac{f}{\tilde{g}})_y,\label{5}\\
	&\frac{D_\tau D_yf\cdot \tilde{g}}{f\tilde{g}}=(\ln f \tilde{g})_{\tau y}+(\ln \frac{f}{\tilde{g}})_\tau (\ln \frac{f}{\tilde{g}})_y.\label{6}
	\end{align}

	Substituting (\ref{1'}) and (\ref{2'}) into (\ref{3'}) and (\ref{4'}),  one obtains
	\begin{align}
	&(\ln f\tilde{g})_{\tau y}+\frac{1}{2\kappa}\frac{\tilde{f}g}{f\tilde{g}}(\ln\frac{\tilde{f}\tilde{g}}{fg})_\tau-\kappa(1-\frac{\tilde{f}g}{f\tilde{g}})=0,\label{3''}\\
	&(\ln \tilde{f}{g})_{\tau y}+\frac{1}{2\kappa}\frac{f\tilde{g}}{\tilde{f}g}(\ln\frac{fg}{\tilde{f}\tilde{g}})_\tau-\kappa(1-\frac{f\tilde{g}}{\tilde{f}g})=0.\label{4''}
	\end{align}
	From (\ref{1'}) and (\ref{2'}), we have
	\begin{align}
	&(\ln \frac{f\tilde{f}}{g\tilde{g}})_y=\frac{1}{2\kappa}(2-\frac{\tilde{f}g}{f\tilde{g}}-\frac{f\tilde{g}}{\tilde{f}g}),\label{a}\\
	&(\ln \frac{fg}{\tilde{f}\tilde{g}})_y=\frac{1}{2\kappa}(\frac{f\tilde{g}}{\tilde{f}g}-\frac{\tilde{f}g}{f\tilde{g}}).\label{b}
	\end{align}
	If we define a new variable by
	\begin{align}
	\phi=\mathrm{i}\ln \frac{f\tilde{g}}{\tilde{f}g},\label{phi}
	\end{align}
	and further define
	$$m = \kappa \tan \phi, \ \ r=\sqrt{m^2+\kappa^2}=\frac{\kappa}{\cos \phi},$$
	then Eqs. (\ref{a}) and (\ref{b}) become
	\begin{align}
	&(\ln \frac{f\tilde{f}}{g\tilde{g}})_y=\frac{1}{\kappa}(1-\cos \phi),\label{a'}\\
	&(\ln \frac{\tilde{f}\tilde{g}}{fg})_{y}=\frac{1}{\mathrm{i}\kappa}\sin\phi.\label{b'}
	\end{align}
	Meanwhile, we have
	\begin{align}
	\frac{\partial x}{\partial y}=\frac{1}{\kappa}+(\ln \frac{g\tilde{g}}{f\tilde{f}})_y=\frac{1}{\kappa}\cos \phi =\frac{1}{r}.\label{r}
	\end{align}
	 Differentiating Eq. (\ref{b'}) with respect to $\tau$ leads to
	\begin{align}
	u_y=\frac{1}{2\mathrm{i}\kappa}(\ln \frac{\tilde{f}\tilde{g}}{fg})_{\tau y}=-\frac{1}{2\kappa^2}(\sin \phi)_\tau,\label{eq1}
	\end{align}
	which, in turn, to be
	\begin{align}
	u_y=-\frac{1}{2\kappa^2}\phi_\tau \cos\phi .\label{tran1}
	\end{align}
	By the definition of $r$, we know that
	\begin{align*}
	-\frac{r_\tau}{r^2}
	=-\frac{1}{\kappa}\phi_\tau \sin\phi.
	\end{align*}
	In other words,
	\begin{align}
	\phi_\tau=\frac{\kappa r_\tau}{r^2\sin\phi}.\label{tran2}
	\end{align}
	Substituting (\ref{tran2}) into (\ref{tran1}), one obtains
	\begin{align}
	r_\tau+2r^2mu_{y}=0\label{eq1'}
	\end{align}
	by the definition of $m$. On the other hand, Using \eqref{3'} and \eqref{4'}, we have
	\begin{align}
	(\ln \frac{f\tilde{g}}{\tilde{f}g})_{\tau y}+\frac{1}{2\kappa}(\frac{\tilde{f}g}{f\tilde{g}}+\frac{f\tilde{g}}{\tilde{f}g})(\ln \frac{\tilde{f}\tilde{g}}{fg})_\tau-\kappa(\frac{f\tilde{g}}{\tilde{f}g}-\frac{\tilde{f}g}{f\tilde{g}})=0.
	\label{c}
	\end{align}
	From (\ref{u}), (\ref{phi}) and (\ref{c}), we obtain
	\begin{align}
		\phi_{\tau y}+2u\cos\phi-2\kappa\sin\phi=0.\label{eq2}
	\end{align}
	Then by using (\ref{eq1}) and (\ref{eq2}), we obtain
	\begin{align}
	-2u\cos^2 \phi+2\kappa\sin\phi\cos\phi=-2\kappa^2u_{yy}-2\kappa^2u_y\phi_y\tan\phi,
	\end{align}
	which is equivalent to
	\begin{align}
	u&=\kappa \tan\phi+\frac{\kappa^2}{\cos^2\phi}u_{yy}+\frac{\kappa}{\cos\phi}\kappa  \phi_y u_y \sec\phi \tan \phi\\
	&=m+r^2u_{yy}+rr_yu_y.
	\end{align}
	Note that
	\begin{align}
	u_{xx}=r(ru_y)_y=r^2u_{yy}+rr_yu_y,
	\end{align}
	so the variable $m$ is expressed as
	\begin{align}
	m=u-u_{xx}.\label{m}
	\end{align}
	Taking use of the expression (\ref{m}), we show that
	\begin{align*}
	x_{\tau y}=\left(\frac{1}{r}\right)_\tau=-\frac{\sin\phi}{\kappa}\phi_\tau=2\kappa u_y\tan\phi=2mu_y=2(u-u_{xx})u_y=2(u-(ru_y)_y)u_y=(u^2-(ru_y)^2)_y,
	\end{align*}
	which implies
	\begin{align}
	x_\tau=u^2-r^2u_y^2=u^2-u_x^2.
	\end{align}
It is consistent with the hodograph transformation (\ref{hodo}). So the mCH equation (\ref{mch}) can be recast into Eqs. (\ref{eq1'}) and (\ref{m}).
\end{proof}
\section{Proof of Lemma 3.2}\label{pro}
\begin{proof}
	Firstly we note that
	\begin{align}
	f_{n+1,k}^{(i)}(l)=f_{n,k+1}^{(i)}(l)+af_{n,k}^{(i)}(l),\,f_{n,k}^{(i)}(l+1)=f_{n,k}^{(i)}(l)-bf_{n+1,k}^{(i)}(l),\,\partial_{x_{-1}}f_{n,k}^{(i)}(l)=f_{n,k-1}^{(i)}(l)
	\end{align}
	from the expression of $f_{n,k}^{(i)}(l)$. We can obtain the following relations
	\begin{align*}
	\tau_{n, k+1}(l)&=|0_{k+1}(l),\cdots,N-1_{k+1}(l)|=|0_{k+1}(l+1),\cdots,N-2_{k+1}(l+1),N-1_{k+1}(l)|\\
	\tau_{n+1, k}(l+1)&=|1_k(l+1),1_{k+1}(l+1),\cdots,N-1_{k+1}(l+1)|\\
	\tau_{n+1, k}(l)&=|1_k(l),1_{k+1}(l),\cdots,N-1_{k+1}(l)|\\
	&=|1_k(l+1),1_{k+1}(l+1),\cdots,N-2_{k+1}(l+1),N-1_{k+1}(l)|+ab\tau_{n+1, k}(l)\\
	\tau_{n, k+1}(l+1)&=|0_{k+1}(l+1),\cdots,N-1_{k+1}(l+1)|\\
	\tau_{n,k}(l+1)&=|0_k(l+1),0_{k+1}(l+1),\cdots,N-2_{k+1}(l+1)|\\
	\tau_{n+1, k+1}(l)&= |1_{k+1}(l),\cdots,N_{k+1}(l)|=|1_{k+1}(l+1),\cdots,N-1_{k+1}(l+1),N_{k+1}(l)|
	\end{align*}
	By Pl{\"u}cker relation, we have
	\begin{align*}
	&\frac{1}{b}(\tau_{n, k+1}(l)\tau_{n+1, k}(l+1)-\tau_{n, k+1}(l+1)\tau_{n+1, k}(l))\\
	=&\frac{1}{b}|0_{k+1}(l+1),\cdots,N-2_{k+1}(l+1),N-1_{k+1}(l)||1_k(l+1),1_{k+1}(l+1),\cdots,N-1_{k+1}(l+1)|\\
	-&\frac{1}{b}|1_k(l+1),1_{k+1}(l+1),\cdots,N-2_{k+1}(l+1),N-1_{k+1}(l)||0_{k+1}(l+1),\cdots,N-1_{k+1}(l+1)|-a\tau_{n+1, k}(l)\tau_{n, k+1}(l+1)\\
	=&\frac{1}{b}|1_{k+1}(l+1),\cdots,N-1_{k+1}(l+1),N-1_{k+1}(l)||1_k(l+1),0_{k+1}(l+1),\cdots,N-2_{k+1}(l+1)|-a\tau_{n+1, k}(l)\tau_{n, k+1}(l+1)\\
	=&a[\tau_{n,k}(l+1)\tau_{n+1, k+1}(l)-\tau_{n+1, k}(l)\tau_{n, k+1}(l+1)],
	\end{align*}
	which is the semi-discrete bilinear equation (\ref{dis1}).
	
	Now we proceed to the proof of equation (\ref{dis2}). Note that
	\begin{align*}
	&D_{x_{-1}}(\frac{1}{b}(\tau_{n, k+1}(l)\cdot \tau_{n+1, k}(l+1)-\tau_{n, k+1}(l+1)\cdot \tau_{n+1, k}(l))+a\tau_{n, k+1}(l+1)\cdot \tau_{n+1, k}(l))\\
	&=\frac{1}{b}D_{x_{-1}}\tau_{n, k+1}(l)\cdot \tau_{n+1, k}(l+1)-D_{x_{-1}}(-\frac{1}{b}\tau_{n, k+1}(l+1)\cdot \tau_{n+1, k}(l)+a\tau_{n, k+1}(l+1)\cdot \tau_{n+1, k}(l))\\
	&=I_3-I_4,
	\end{align*}
	where
	\begin{align*}
	I_3&=\frac{1}{b}D_{x_{-1}}\tau_{n, k+1}(l)\cdot \tau_{n+1, k}(l+1)\\
	&=\frac{1}{b}D_{x_{-1}}|0_{k+1}(l+1),\cdots,N-2_{k+1}(l+1),N-1_{k+1}(l)|\cdot |1_k(l+1),1_{k+1}(l+1),\cdots,N-1_{k+1}(l+1)|\\
	I_4&=D_{x_{-1}}(-\frac{1}{b}\tau_{n, k+1}(l+1)\cdot \tau_{n+1, k}(l)+a\tau_{n, k+1}(l+1)\cdot \tau_{n+1, k}(l))\\
	&=-\frac{1}{b}D_{x_{-1}}|0_{k+1}(l+1),\cdots,N-1_{k+1}(l+1)|\cdot|1_k(l+1),1_{k+1}(l+1),\cdots,N-2_{k+1}(l+1),N-1_{k+1}(l)|
	\end{align*}
	One can verify that
	\begin{align*}
	\partial_{x_{-1}}|0_{k+1}(l+1),\cdots,N-2_{k+1}(l+1),N-1_{k+1}(l)|=&\sum_{i=0}^{N-2}|0_{k+1}(l+1),\cdots,i_k(l+1),\cdots,N-2_{k+1}(l+1),N-1_{k+1}(l)|\\
	&+|0_{k+1}(l+1),\cdots,N-2_{k+1}(l+1),N-1_{k}(l)|\\
	\partial_{x_{-1}}|1_k(l+1),1_{k+1}(l+1),\cdots,N-1_{k+1}(l+1)|=&|1_{k-1}(l+1),1_{k+1}(l+1),\cdots,N-1_{k+1}(l+1)|\\
	\partial_{x_{-1}}|0_{k+1}(l+1),\cdots,N-1_{k+1}(l+1)|=&\sum_{i=0}^{N-1}|0_{k+1}(l+1),\cdots,i_k(l+1),\cdots,N-1_{k+1}(l+1)|\\
	\partial_{x_{-1}}|1_k(l+1),1_{k+1}(l+1),\cdots,N-1_{k+1}(l)|=&|1_{k-1}(l+1),1_{k+1}(l+1),\cdots,N-2_{k+1}(l+1),N-1_{k+1}(l)|\\
	&+|1_{k}(l+1),1_{k+1}(l+1),\cdots,N-2_{k+1}(l+1),N-1_{k}(l)|\\
	\partial_{x_{-1}}|1_{k+1}(l+1),\cdots,N-1_{k+1}(l+1),N_{k+1}(l)|=&\sum_{i=1}^{N-1}|1_{k+1}(l+1),\cdots,i_k(l+1),\cdots,N-1_{k+1}(l+1),N_{k+1}(l)|\\
	&+|1_{k+1}(l+1),\cdots,N-1_{k+1}(l+1),N_{k}(l)|\\
	\partial_{x_{-1}}|0_k(l+1),0_{k+1}(l+1),\cdots,N-2_{k+1}(l+1)|=&|0_{k-1}(l+1),0_{k+1}(l+1),\cdots,N-2_{k+1}(l+1)|.
	\end{align*}
	From the following Pl{\"u}cker relations for determinants
	\begin{align*}
	&|0_{k+1}(l+1),\cdots,i_k(l+1),\cdots,N-2_{k+1}(l+1),N-1_{k+1}(l)||1_k(l+1),1_{k+1}(l+1),\cdots,N-1_{k+1}(l+1)|\\
	&-|0_{k+1}(l+1),\cdots,i_k(l+1),\cdots,N-1_{k+1}(l+1)||1_k(l+1),1_{k+1}(l+1),\cdots,N-2_{k+1}(l+1),N-1_{k+1}(l)|\\
	=&ab|1_{k+1}(l+1),\cdots,i_k(l+1),\cdots,N-1_{k+1}(l+1),N_{k+1}(l)||0_k(l+1),0_{k+1}(l+1),\cdots,N-2_{k+1}(l+1)|,i=1,\cdots,N-2\\
\\
	&-|0_{k+1}(l+1),\cdots,N-2_{k+1}(l+1),N-1_{k+1}(l)||1_{k-1}(l+1),1_{k+1}(l+1),\cdots,N-1_{k+1}(l+1)|\\
	&+|0_{k+1}(l+1),\cdots,N-1_{k+1}(l+1)||1_{k-1}(l+1),1_{k+1}(l+1),\cdots,N-2_{k+1}(l+1),N-1_{k+1}(l)|\\
	&=-|1_{k+1}(l+1),\cdots,N-1_{k+1}(l+1),N-1_{k+1}(l)||1_{k-1}(l+1),0_{k+1}(l+1),\cdots,N-2_{k+1}(l+1)|\\
	&=-ab|1_{k+1}(l+1),\cdots,N-1_{k+1}(l+1),N_{k+1}(l)||0_{k-1}(l+1),0_{k+1}(l+1),1_{k+1}(l+1),\cdots,N-2_{k+1}(l+1)|\\
	&-b|1_{k+1}(l+1),\cdots,N-1_{k+1}(l+1),N_{k+1}(l)||0_k(l+1),0_{k+1}(l+1),1_{k+1}(l+1),\cdots,N-2_{k+1}(l+1)|,\\
\\
	&|0_{k}(l+1),\cdots,N-2_{k+1}(l+1),N-1_{k+1}(l)||1_k(l+1),1_{k+1}(l+1),\cdots,N-1_{k+1}(l+1)|\\
	&-|0_{k}(l+1),\cdots,N-1_{k+1}(l+1)||1_{k}(l+1),1_{k+1}(l+1),\cdots,N-2_{k+1}(l+1),N-1_{k+1}(l)|\\
	&=-|1_{k+1}(l+1),\cdots,N-1_{k+1}(l+1),N-1_{k+1}(l)||0_{k}(l+1),1_{k}(l+1),\cdots,N-2_{k+1}(l+1)|\\
	&=-b|1_{k+1}(l+1),\cdots,N-1_{k+1}(l+1),N_{k+1}(l)||0_{k}(l+1),0_{k+1}(l+1),\cdots,N-2_{k+1}(l+1)|
	\end{align*}
	we just need to verify
	\begin{align*}
	&\frac{1}{b}|0_{k+1}(l+1),\cdots,N-2_{k+1}(l+1),N-1_{k}(l)||1_k(l+1),1_{k+1}(l+1),\cdots,N-1_{k+1}(l+1)|\\
	&-\frac{1}{b}|0_{k+1}(l+1),\cdots,N-1_{k}(l+1)||1_k(l+1),1_{k+1}(l+1),\cdots,N-2_{k+1}(l+1),N-1_{k+1}(l)|\\
	&+\frac{1}{b}|0_{k+1}(l+1),\cdots,N-1_{k+1}(l+1)||1_k(l+1),1_{k+1}(l+1),\cdots,N-2_{k+1}(l+1),N-1_k(l)|\\
	&-2|0_{k+1}(l+1),\cdots,N-1_{k+1}(l+1)||1_k(l),1_{k+1}(l+1),\cdots,N-2_{k+1}(l+1),N-1_{k+1}(l)|\\
	&=a|1_{k+1}(l+1),\cdots,N-1_{k}(l+1),N_{k+1}(l)||0_k(l+1),0_{k+1}(l+1),\cdots,N-2_{k+1}(l+1)|\\
	&+a|1_{k+1}(l+1),\cdots,N-1_{k+1}(l+1),N_{k}(l)||0_k(l+1),0_{k+1}(l+1),\cdots,N-2_{k+1}(l+1)|.
	\end{align*}
	Since
	\begin{align*}
	&-\frac{1}{b}|0_{k+1}(l+1),\cdots,N-1_{k}(l+1)||1_k(l+1),1_{k+1}(l+1),\cdots,N-2_{k+1}(l+1),N-1_{k+1}(l)|\\
	&=-\frac{1}{b}|0_{k+1}(l+1),\cdots,N-2_{k+1}(l+1),N-1_{k+1}(l)||1_k(l+1),1_{k+1}(l+1),\cdots,N-2_{k+1}(l+1),N-1_k(l+1)|\\
	&+\frac{1}{b}|1_{k+1}(l+1),\cdots,N-2_{k+1}(l+1),N-1_k(l+1),N-1_{k+1}(l)||1_k(l+1),0_{k+1}(l+1),\cdots, N-2_{k+1}(l+1)|\\
	&=a|1_{k+1}(l+1),\cdots,N-1_{k}(l+1),N_{k+1}(l)||0_k(l+1),0_{k+1}(l+1),\cdots,N-2_{k+1}(l+1)|\\
	&+\frac{a}{b}|1_{k+1}(l+1),\cdots,N-2_{k+1}(l+1),N-1_k(l+1),N-1_{k+1}(l+1)||0_k(l+1),0_{k+1}(l+1),\cdots, N-2_{k+1}(l+1)|\\
\\
	&\frac{1}{b}|0_{k+1}(l+1),\cdots,N-2_{k+1}(l+1),N-1_{k}(l)||1_k(l+1),1_{k+1}(l+1),\cdots,N-1_{k+1}(l+1)|\\
	&-\frac{1}{b}|0_{k+1}(l+1),\cdots,N-1_{k+1}(l+1)||1_k(l+1),1_{k+1}(l+1),\cdots,N-2_{k+1}(l+1),N-1_k(l)|\\
	&=\frac{a}{b}|1_{k+1}(l+1),\cdots,N-1_{k+1}(l+1),N-1_k(l)||0_k(l+1),0_{k+1}(l+1),\cdots,N-2_{k+1}(l+1)|\\
	&=a|1_{k+1}(l+1),\cdots,N-1_{k+1}(l+1),N_k(l)||0_k(l+1),0_{k+1}(l+1),\cdots,N-2_{k+1}(l+1)|\\
	&+\frac{a}{b}|1_{k+1}(l+1),\cdots,N-1_{k+1}(l+1),N-1_k(l+1)||0_k(l+1),0_{k+1}(l+1),\cdots,N-2_{k+1}(l+1)|,
	\end{align*}
	the equation to be verified can be transformed into
	\begin{align*}
	&\frac{2}{b}|0_{k+1}(l+1),\cdots,N-1_{k+1}(l+1)||1_k(l+1),1_{k+1}(l+1),\cdots,N-2_{k+1}(l+1),N-1_k(l)|\\
	&-2|0_{k+1}(l+1),\cdots,N-1_{k+1}(l+1)||1_k(l),1_{k+1}(l+1),\cdots,N-2_{k+1}(l+1),N-1_{k+1}(l)|\\
	&=0,
	\end{align*}
	which is equivalent to
	\begin{align*}
	\frac{1}{b}|1_k(l+1),1_{k+1}(l+1),\cdots,N-2_{k+1}(l+1),N-1_k(l)|-|1_k(l),1_{k+1}(l+1),\cdots,N-2_{k+1}(l+1),N-1_{k+1}(l)|=0.
	\end{align*}
	On the other hand, we notice that
	\begin{align*}
	&\frac{1}{b}|1_k(l+1),1_{k+1}(l+1),\cdots,N-2_{k+1}(l+1),N-1_k(l)|\\
	=&\frac{1}{b^2}|1_k(l+1),1_{k+1}(l+1),\cdots,N-2_{k+1}(l+1),N-2_k(l)|\\
	=&\cdots\\
	=&\left(\frac{1}{b}\right)^{N-2}|1_k(l),1_{k+1}(l+1),\cdots,N-2_{k+1}(l+1),2_k(l)|,\\
	\\
	&|1_k(l),1_{k+1}(l+1),\cdots,N-2_{k+1}(l+1),N-1_{k+1}(l)|\\
	=&\frac{1}{b}|1_k(l),1_{k+1}(l+1),\cdots,N-2_{k+1}(l+1),N-2_{k+1}(l)|\\
	=&\cdots\\
	=&\left(\frac{1}{b}\right)^{N-2}|1_k(l),1_{k+1}(l+1),\cdots,N-2_{k+1}(l+1),2_{k}(l)|.
	\end{align*}
	Thus we have proven that the Casorati determinant solution satisfies the bilinear equations (\ref{dis1}) and (\ref{dis2}).
\end{proof}

\end{document}